\documentclass[letterpaper, 10 pt]{ieeetran}  % Comment this line out if you need a4paper  
\IEEEoverridecommandlockouts                              % This command is only needed if you want to use the \thanks command%\documentclass[journal,twoside,web]{ieeecolor}
\usepackage{graphics} % for pdf, bitmapped graphics files
\usepackage{epsfig} % for postscript graphics files
\usepackage{amsmath} % assumes amsmath package installed
\usepackage{amssymb}  % assumes amsmath package installed
 \usepackage{enumitem}

\usepackage{graphicx}

\usepackage{amsthm}

\newtheorem{assumption}{\bf Assumption}
\newtheorem{definition}{\bf Definition}
\newtheorem{theorem}{\bf Theorem}
\newtheorem{proposition}{\bf Proposition}

\newtheorem{corollary}{\bf Corollary}
\newtheorem{remark}{\bf Remark}

\newcommand{\prob}[1]{\mathrm{Pr}\left[#1\right]}
\newcommand{\tr}[1]{\mathrm{tr}\left(#1\right)}

 \IEEEpubidadjcol 
\IEEEoverridecommandlockouts

\usepackage{setspace}

\usepackage{tikz}
 \usetikzlibrary{arrows,patterns,mindmap,backgrounds,shadows}
\usetikzlibrary{backgrounds}
\usetikzlibrary{shapes,arrows,chains}

\usepackage{siunitx}

\usepackage{cite}
\usepackage{hyperref}
\usepackage{comment}

\usepackage{algpseudocode, algorithm}
\usepackage{enumitem}

\usepackage{multirow}

\begin{document}
\title{Predictive control for nonlinear stochastic systems:\\
Closed-loop guarantees with unbounded noise}
\author{Johannes K\"ohler, Melanie N. Zeilinger%
\thanks{Institute for Dynamic Systems and Control, ETH Zürich, Zürich CH-8092, Switzerland (email:$\{$jkoehle$|$mzeilinger$\}$@ethz.ch).}
\thanks{Johannes K\"ohler was supported by the Swiss National Science Foundation under NCCR Automation (grant agreement 51NF40 180545).}}
\IEEEoverridecommandlockouts
\IEEEpubid{\begin{minipage}{\textwidth}\ \\[20pt] \\ \\
         \copyright 2025 IEEE.  Personal use of this material is  permitted.  Permission from IEEE must be obtained for all other uses, in  any current or future media, including reprinting/republishing this material for advertising or promotional purposes, creating new  collective works, for resale or redistribution to servers or lists, or  reuse of any copyrighted component of this work in other works.
     \end{minipage}}
\maketitle
%%%%%%%%%%%%%%%%%%%%%%%%%%%%%%%%%%%%%%%%%%%%%%%%%%%%%%%%%%%%%%%%%%%%%%%%%%%%%%%%
\begin{abstract}
%!TEX root = main.tex
%%%%%%%%%%%%%%%%%%%%%%%%%%%%%%%%%%%%%%%%%%%%%%%%%%%%%%%%%%%%%%%%%%%%%%%%%%%%%%% 
We present a stochastic model predictive control framework for nonlinear systems subject to unbounded  process noise with closed-loop guarantees.
First, we provide a conceptual shrinking-horizon framework that utilizes general probabilistic reachable sets and minimizes the expected cost. 
Then, we provide a tractable receding-horizon formulation that uses a nominal state to minimize a deterministic quadratic cost and satisfy tightened constraints. 
Our theoretical analysis demonstrates recursive feasibility, satisfaction of chance constraints, and bounds on the expected cost for the resulting closed-loop system. 
We provide a constructive design for probabilistic reachable sets of nonlinear continuously differentiable systems using stochastic contraction metrics and an assumed bound on the covariance matrices. 
Numerical simulations highlight the computational efficiency and theoretical guarantees of the proposed method.
Overall, this paper provides a framework for computationally tractable stochastic predictive control with closed-loop guarantees for nonlinear systems with unbounded noise. 

\end{abstract}  
%!TEX root = main.tex
\section{Introduction}
Model predictive control (MPC) is an optimization-based method that yields high-performance control for general nonlinear systems and ensures satisfaction of safety-critical constraints~\cite{rawlings2017model}. 
Robust MPC methods predict over-approximations of the robust reachable set to ensure constraint satisfaction for any admissible bounded noise~\cite{Kouvaritakis2016textbook}. 
However, such worst-case bounds can be overly conservative or simply inadequate if no bound on the noise is available. 
Stochastic MPC (SMPC) formulations avoid these problems by leveraging distributional information about the noise and allowing for a user-chosen probability of constraint violation~\cite{mesbah2016stochastic}. 
In this paper, we address the design of SMPC schemes that yield suitable closed-loop guarantees for nonlinear systems subject to unbounded noise. 

\subsection{Related work} 
The design of SMPC schemes with closed-loop guarantees faces two key challenges:
\begin{enumerate}
\item How to reformulate chance constraints as tractable deterministic conditions?
\item How to formulate a tractable finite-horizon problem that yields closed-loop guarantees?
\end{enumerate}
In the following, we discuss state-of-the-art methods addressing these two challenges.\\
\textit{Probabilistic reachable sets for nonlinear systems:} 
The reformulation of chance constraints can be equivalently posed as the computation of probabilistic reachable sets (PRS), i.e., sets which contain uncertain future states at least with a specified probability~\cite[Def.~4]{Hewing2018recovery_mechanism}. 
For linear stochastic systems, such PRS can be efficiently computed using analytical bounds~\cite{cinquemani2011convexity,Farina2016soverview,paulson2020stochastic,Hewing2018recovery_mechanism,Hewing2020indirect} or offline sampling~\cite{hewing2019scenario,Lorenzen2017tightening,aolaritei2023wasserstein,campi2019scenario}. 
For nonlinear stochastic systems, there exist various approaches to approximate PRS, such as (generalized) polynomial chaos expansion~\cite{paulson2019efficient,bradford2019output}, sampling~\cite{kantas2009sequential,sehr2017particle}, or linearization~\cite{hewing2019cautious}, see also the overviews~\cite{mesbah2016stochastic,landgraf2023probabilistic}. 
These methods typically trade off computational complexity with approximation errors and thus probabilistic containment is not guaranteed. 
In case of bounded noise, there exist methods to compute valid over-approximations of the robust reachable set~\cite{mayne2011tube,cannon2011robust,houska2019robust}, such as contraction metrics~\cite{zhao2022tube,singh2023robust,sasfi2023robust}. 
Similarly, valid PRS can be effectively computed for nonlinear stochastic systems if robust bounds on the noise are available~\cite{bonzanini2019tube,schluter2020constraint,zhong2023nonlinear}. 
The computation of valid PRS for nonlinear systems subject to unbounded stochastic noise remains a challenging problem.

\textit{Closed-loop properties in SMPC:} 
A key challenge in the design of SMPC schemes is that constraint violations are explicitly permitted with some non-zero probability. Hence, na\"ive implementations may lose feasibility and thus all closed-loop properties during online operation~\cite[Sec.~IV.D]{primbs2009stochastic}. 
Two conceptual SMPC frameworks to address this problem are \textit{robust techniques} and \textit{feasibility-preserving algorithms} (cf.~\cite[Sec.~I]{koehler2023stochastic}). 
Robust techniques impose more restrictive constraints to robustly ensure recursive feasibility for the worst-case noise realization. 
Corresponding linear and nonlinear SMPC schemes can be found in~\cite{cannon2010stochastic,korda2011strongly,Kouvaritakis2016textbook,Lorenzen2017tightening,aolaritei2023wasserstein} and \cite{lorenzen2019stochastic,schluter2020constraint,mcallister2022nonlinear}, respectively. 
Key limitations include the requirement of a known worst-case bound on the noise and conservatism (cf. comparison in~\cite[Sec.~IV]{koehler2023stochastic}). 
In contrast, feasibility-preserving algorithms specify the initial condition in the SMPC formulation such that closed-loop properties are preserved independently of the realized noise~\cite{Farina2016soverview,Hewing2018recovery_mechanism,Hewing2020indirect,schluter2022stochastic,Koehler2022interpolating,wang2021recursive}. 
One particularly promising approach is indirect-feedback SMPC~\cite{Hewing2020indirect}, which has been the basis for many recent extensions and developments~\cite{hewing2019scenario,muntwiler2022lqg,mark2021data,arcari2023stochastic}. 
This approach leverages linear system dynamics and a nominal state initialization to define stochastic error dynamics that evolve completely independent of the variables optimized by the SMPC. 
As a result, chance constraints on the closed-loop system can be efficiently formulated as tightened deterministic constraints on the nominal state predicted by the SMPC. 
However, the developments in~\cite{Hewing2020indirect,hewing2019scenario,muntwiler2022lqg,mark2021data,arcari2023stochastic} strongly rely on the independence between the error and the nominal trajectory computed by the SMPC, which prohibits application to nonlinear systems. 
In~\cite{wang2021recursive}, a feasibility-preserving SMPC for nonlinear systems is proposed that uses online adjustments on the probability level in the constraints to obtain closed-loop guarantees. 
However, application requires online evaluation of probability levels through sampling, which results in a significant increase in computational complexity and limits application to short finite-horizon problems. 
Overall, the design of tractable nonlinear SMPC schemes with closed-loop guarantees remains largely unsolved~\cite{mayne2016robust}.

\subsection{Contribution}
We provide a computationally tractable framework for nonlinear SMPC that yields closed-loop guarantees for unbounded noise. 
This is based on three main technical contributions: 
\begin{itemize}
\item We generalize the indirect-feedback SMPC framework~\cite{Hewing2020indirect} to nonlinear systems by removing the independence assumption and using general PRS; 
\item We provide a tractable SMPC scheme using a nominal system and tightened constraints (Thm.~\ref{thm:SMPC}); 
\item We extend results on stochastic contraction metrics~\cite{pham2013stochastic,tsukamoto2020robust} to design PRS for nonlinear stochastic systems with unbounded noise (Thm.~\ref{thm:CCM}--\ref{thm:PRS}). 
\end{itemize}
The resulting nonlinear SMPC scheme has a computational demand comparable to a nominal MPC scheme and ensures the following closed-loop properties: 
\begin{enumerate}[label=P\arabic*)]
\item Recursive feasibility;
\label{enum_desired_feas}
\item Chance constraint satisfaction;
\label{enum_desired_chance}
\item Bound on expected cost. 
\label{enum_desired_cost}
\end{enumerate} 
The proposed design is applicable to independent zero mean noise with a known bound on the covariance, quadratic costs, continuously differentiable dynamics that are linear in the noise, and continuously differentiable constraints. 
We demonstrate the practical applicability and computational efficiency of the proposed design with a numerical example involving a chain of mass-spring-dampers with nonlinear Coulomb friction.

\subsubsection*{Outline} 
We first present the problem setup (Sec.~\ref{sec:setup}) and propose a conceptual framework for SMPC with a shrinking-horizon formulation (Sec.~\ref{sec:SNMPC_shrinking}). 
Then, we provide a tractable formulation with nominal predictions (Sec.~\ref{sec:SNMPC_tractable}) 
and show how to compute PRS for nonlinear systems using stochastic contraction metrics (Sec.~\ref{sec:CCM}). 
Afterwards, we summarize the overall design (Sec.~\ref{sec:algorithm}) and provide a discussion (Sec.~\ref{sec:discussion}). 
Finally, we illustrate the results with a numerical example (Sec.~\ref{sec:num}) and end with a conclusion (Sec.~\ref{sec:conclusion}). 
Appendix~\ref{app_CCM} details the offline design of contraction metrics and Appendix~\ref{app:affine_w} considers dynamics that are non-affine in the noise.

\subsubsection*{Notation}
The set of integers in an interval $[a,b]$ is denoted by $\mathbb{I}_{[a,b]}$. 
By $u(a:b)\in\mathbb{U}^{b-a+1}$, $a,b\in\mathbb{I}_{\geq 0}$, we denote the sequence with elements $u(k)\in\mathbb{U}$, $k\in\mathbb{I}_{[a,b]}$. 
We denote the prediction for time $i$ computed at time $k$ by $\mathbf{u}_{i|k}\in\mathbb{U}$ and the predicted sequence containing elements $i\in\mathbb{I}_{[a,b]}$ by $\mathbf{u}_{a:b|k}\in\mathbb{U}^{b-a+1}$. 
Whenever clear from the context, we denote the full predicted sequence by $\mathbf{u}_{\cdot|k}$. 
By $Q\succ 0$ ($\succeq 0$) we denote that a symmetric matrix $Q$ is positive (semi)-definite and by $Q^{1/2}$ we denote the symmetric matrix square-root, i.e., $Q^{1/2}Q^{1/2}=Q$. 
We denote the Euclidean norm of a vector $x$ by $\|x\|=\sqrt{x^\top x}$ and the weighted norm w.r.t. a positive definite matrix $M$ by $\|x\|_M:=\sqrt{x^\top M x}$. 
For two sets $\mathbb{A},\mathbb{B}\subseteq\mathbb{R}^n$, the Minkowski sum is defined a $\mathbb{A}\oplus\mathbb{B}=\{a+b\mid a\in\mathbb{A},~b\in\mathbb{B}\}$. 
The probability of an event $A$ is denoted by $\prob{A}$. 
The expectation of a function $\delta(w)$ over a random variable $w$ is denoted by $\mathbb{E}_w[\delta(w)]$ and the expectation conditioned on an event $A$ is given by $\mathbb{E}_w[\delta(w)\mid A]$. 
By $\mathcal{K}_\infty$, we denote the set of continuous functions $\alpha:\mathbb{R}_{\geq 0}\rightarrow\mathbb{R}_{\geq 0}$ that are strictly increasing, unbounded, and satisfy $\alpha(0)=0$.
 For a continuously differentiable function $f:\mathbb{R}^a \times \mathbb{R}^b \to \mathbb{R}^c$, $f(x,u)$, the partial derivative w.r.t. $x$ evaluated at some point $(x,u) = (z,v)$ is defined as $\left.\frac{\partial f}{\partial x}\right|_{(z,v)} \in \mathbb{R}^{c \times a}$. 

%!TEX root = main.tex
\section{Problem formulation}
\label{sec:setup}
We consider a nonlinear stochastic system
\begin{align}
\label{eq:sys}
x(k+1)=f(x(k),u(k),w(k)),\quad x(0)=x_0,
\end{align}
with state $x(k)\in\mathbb{R}^n$, input $u(k)\in\mathbb{U}\subseteq\mathbb{R}^m$, process noise $w(k)\in\mathbb{R}^q$, discrete time $k\in\mathbb{I}_{\geq 0}$, initial condition $x_0\in\mathbb{R}^n$, and input constraint $\mathbb{U}$. 
The dynamics $f$ are known and the state $x(k)$ can be measured. 
\begin{assumption}(Stochastic noise)
\label{asm:disturbance}
The noise $w(k)$ is independently distributed according to distributions $\mathcal{Q}_{\mathrm{w}}(k)$ with zero mean and a variance bound $\Sigma_{\mathrm{w}}\succ 0$, i.e., $\mathbb{E}_{w(k)}[w(k)]=0$, $\mathbb{E}_{w(k)}[w(k)w(k)^\top]\preceq \Sigma_{\mathrm{w}}$, $\forall k\in\mathbb{I}_{\geq 0}$. 
\end{assumption}
Assumption~\ref{asm:disturbance} is quite general, allowing for time-varying, non-Gaussian, and unbounded distributions with a known bound $\Sigma_{\mathrm{w}}$ on their variance. The proposed design in Section~\ref{sec:algorithm} will be distributionally robust, i.e., we provide guarantees that are valid for any distribution satisfying Assumption~\ref{asm:disturbance} and implementation does not require knowledge of the distribution $\mathcal{Q}_{\mathrm{w}}$. 
We impose chance constraints of the form 
\begin{align}
\label{eq:prob_state}
\prob{x(k)\in\mathbb{X}}\geq p\quad \forall k\in\mathbb{I}_{\geq 0},
\end{align}
where $\mathbb{X}\subseteq\mathbb{R}^n$ is some closed set and $p\in(0,1)$ is a desired probability level. 
We define the finite-horizon cost 
\begin{align*}
&\mathcal{J}_N(x(0:N),u(0:N-1))\\:=
&\sum_{k=0}^{N-1}\ell(x(k),u(k))+V_{\mathrm{f}}(x(N)),
\end{align*}
with a user chosen stage cost $\ell:\mathbb{R}^n\times\mathbb{U}\rightarrow\mathbb{R}$, a terminal cost $V_{\mathrm{f}}:\mathbb{R}^n\rightarrow\mathbb{R}$, and a horizon $N\in\mathbb{I}_{\geq 1}$.
We assume that $f,\ell,V_{\mathrm{f}}$ are continuous and $\mathbb{U}$ is compact. 
Ideally, we would like to solve the following stochastic optimal control problem 
\begin{align}
\label{eq:SOCP}
\inf_{\pi}~&\lim_{\bar{N}\rightarrow\infty}\dfrac{1}{\bar{N}}\mathbb{E}_{w(0:\bar{N}-1})\left[\mathcal{J}_{\bar{N}}(x(0:\bar{N}),u(0:\bar{N}-1))\right]\\
\text{s.t. }&\eqref{eq:sys},~\eqref{eq:prob_state},~u(k)=\pi_k(w(0:k-1))\in\mathbb{U},~
k\in\mathbb{I}_{\geq 0},\nonumber
\end{align}
where $\pi$ are causal policies that minimize the infinite-horizon expected cost and ensure satisfaction of the chance constraints~\eqref{eq:prob_state} $\forall k\in\mathbb{I}_{\geq 0}$. 
Problem~\eqref{eq:SOCP} is not computationally tractable for multiple reasons:
(i) optimization over policies $\pi$; (ii) infinite prediction horizon, and (iii) chance constraints~\eqref{eq:prob_state}. 
In this paper, we derive a computational tractable SMPC scheme that uses a receding-horizon implementation, optimizes open-loop inputs, and uses PRS to ensure satisfaction of the chance constraints. 
For some of the specific designs (Sec.~\ref{sec:algorithm}) and the performance analysis (Sec.~\ref{sec:SNMPC_tractable}), we will consider the following standard regularity conditions.
\begin{assumption}(Regularity conditions)
\label{asm:setup_regularity}
\begin{enumerate}[label=\roman*)]
\item The dynamics $f$ are Lipschitz continuous with respect to all arguments. 
\label{asm:setup_regularity_dynamics}
\item The cost is given by $\ell(x,u)=\|x\|_Q^2+\|u\|_R^2$, $V_{\mathrm{f}}(x)=\|x\|_P^2$ with positive definite matrices $Q,R,P$.
\label{asm:setup_regularity_cost}
\item The origin is an equilibrium, i.e., $f(0,0,0)=0$, $0\in\mathbb{U}$, and $0\in\mathrm{interior}(\mathbb{X})$. 
\label{asm:setup_regularity_origin}
\item The state constraint is given by $\mathbb{X}=\left\{x\in\mathbb{R}^n|~h_j(x)\leq 0,~j\in\mathbb{I}_{[1,r]}\right\}$ with $h_j:\mathbb{R}^n\rightarrow\mathbb{R}$ continuously differentiable, Lipschitz continuous, and (w.l.o.g.) $h_j(0)=-1$.
\label{asm:setup_regularity_constraint}
\item The dynamics $f$ are continuously differentiable with respect to all arguments and linear in $w$, i.e.,\\ $f(x,u,w)\equiv f(x,u,0)+Gw$, $G\in\mathbb{R}^{n\times q}$.
\label{asm:setup_regularity_CCM}
\end{enumerate}
\end{assumption}
\begin{remark}
\label{rk:linear_noise}(Linearity in noise)
Condition~\ref{asm:setup_regularity_CCM} requires that the noise $w$ appears linearly in the dynamics, which facilitates the construction of probabilistic reachable sets in Section~\ref{sec:CCM}. In case $w$ enters non-linearly, the construction of the reachable set can be formulated in an augmented state $[x(k);w(k)]\in\mathbb{R}^{n+p}$, see Appendix~\ref{app:affine_w} for details. 
\end{remark}
\begin{remark}
\label{rk:input}(Probabilistic input constraints)
In SMPC, input constraints $u(k)\in\mathbb{U}$ are often relaxed to probabilistic input constraints $\prob{u(k)\in\mathbb{U}}\geq p$~\cite{Hewing2018recovery_mechanism,Hewing2020indirect,schluter2022stochastic,Koehler2022interpolating}. 
We consider hard input constraints due to their prevalence in practical applications and to simplify the exposition of PRS. 
However, the presented results can be naturally adjusted to this setup, see Section~\ref{sec:discussion_PRS} for details. 
\end{remark}

%!TEX root = main.tex
\section{SMPC using PRS - the shrinking-horizon case}
\label{sec:SNMPC_shrinking}
In this section, we present the proposed framework for nonlinear SMPC using general PRS. 
We first consider a shrinking-horizon problem with some finite horizon $N\in\mathbb{I}_{>0}$ and treat the more general receding-horizon problem in Section~\ref{sec:SNMPC_tractable}. 
We focus on how to incorporate the state measurement $x(k)$ in the SMPC and the resulting closed-loop properties. 
To this end, we consider the following definition of PRS.
\begin{definition}[PRS]
\label{def:PRS}
 Consider system~\eqref{eq:sys} in closed loop with any causal policy $u(k)=\pi_k(w(0:k-1))\in\mathbb{U}$, $k\in\mathbb{I}_{\geq 0}$. 
For a given probability $p\in(0,1)$, a sequence of sets $\mathcal{R}_k$, $k\in\mathbb{I}_{\geq 0}$ are probabilistic reachable sets (PRS) if
\begin{align}
\label{eq:PRS}
\prob{x(k)\in\mathcal{R}_k}\geq p, \quad \forall k\in\mathbb{I}_{\geq 0}.
\end{align}
\end{definition}
In this paper, we focus on optimizing open-loop input sequences and consider the following parametrization of PRS.
\begin{assumption}[PRS]
\label{asm:PRS}
We know a sequence of closed parametrized sets $\mathcal{R}_k:\mathbb{R}^n\times\mathbb{U}^k\rightarrow2^{\mathbb{R}^n}$, $k\in\mathbb{I}_{\geq 0}$, such that 
$\mathcal{R}_k(x_0,u(0:k-1))$, $k\in\mathbb{I}_{\geq 0}$, are PRS (Definition~\ref{def:PRS}). 
\end{assumption}
The constructive design of such PRS is studied in Section~\ref{sec:CCM}. 
At each time $k\in\mathbb{I}_{[0,N-1]}$, the proposed shrinking-horizon SMPC considers the following optimization problem:
\begin{subequations}
\label{eq:SNMPC_shrinking}
\begin{align}
\label{eq:SNMPC_shrinking_cost}
\inf_{\mathbf{u}_{0:N-1|k}\in\mathbb{U}^{N}}&~\mathbb{E}_{w(k:N-1)}\left[\mathcal{J}_{N-k}(\mathbf{x}_{k:N|k},\mathbf{u}_{k:N-1|k})\right]\\
\label{eq:SNMPC_shrinking_state}
\text{s.t. }&\mathcal{R}_i(x_0,\mathbf{u}_{0:i-1|k})\subseteq \mathbb{X},\\
\label{eq:SNMPC_shrinking_input}
&\mathbf{u}_{0:k-1|k}={u}(0:k-1),\\
\label{eq:SNMPC_shrinking_init}
&\mathbf{x}_{k|k}=x(k),\\
\label{eq:SNMPC_shrinking_prediction}
&\mathbf{x}_{i+1|k}=f(\mathbf{x}_{i|k},\mathbf{u}_{i|k},w(i)), \\
\label{eq:SNMPC_shrinking_prediction_w}
&w(i)\sim\mathcal{Q}_{\mathrm{w}}(i),\\
&i\in\mathbb{I}_{[k,N-1]},\nonumber
\end{align}
\end{subequations}
which depends on the current state $x(k)$ and the past applied inputs ${u}(0:k-1)$. 
We assume that a minimizing input sequence exists\footnote{%
If $f,\ell,V_{\mathrm{f}}$ are \textit{uniformly} continuous, then the expected cost is a well-defined continuous function of $\mathbf{u}_{\cdot|k}$, see~\cite[Lemma 16]{mcallister2022nonlinear}. 
In this case, compact constraints $\mathbb{U}$ and the closed constraints~\eqref{eq:SNMPC_shrinking_state} ensure that a minimizer exists for each $x(k)\in\mathbb{R}^n$, $k\in\mathbb{I}_{[0,N-1]}$, assuming the problem is feasible.}, if the problem is feasible. We denote a minimizing input sequence by $\mathbf{u}^\star_{\cdot|k}\in\mathbb{U}^{N}$ and the corresponding optimal predicted state distribution by $\mathbf{x}_{k:N|k}^\star$. 
In closed-loop operation, we solve Problem~\eqref{eq:SNMPC_shrinking} at each time $k\in\mathbb{I}_{[0,N-1]}$ and apply the optimized input for this point in time, i.e., $u(k)=\mathbf{u}^\star_{k|k}$. 
This formulation is insipred by the linear SMPC scheme in~\cite{Hewing2020indirect}, since the state measurement $x(k)$ is only used for the expected cost through the stochastic prediction~\eqref{eq:SNMPC_shrinking_init}--\eqref{eq:SNMPC_shrinking_prediction_w}, 
while the chance constraints~\eqref{eq:prob_state} are enforced through the constraints~\eqref{eq:SNMPC_shrinking_state}, which do not directly dependent on the past random noise realizations $w(0:k-1)$. 
The constraint~\eqref{eq:SNMPC_shrinking_input} ensures that the MPC considers a causal policy. 
The following theorem formalizes the closed-loop properties of the proposed shrinking-horizon SMPC formulation. 
\begin{theorem}
\label{thm:SMPC_shrinking}
Let Assumptions~\ref{asm:disturbance} and \ref{asm:PRS} hold and suppose Problem~\eqref{eq:SNMPC_shrinking} is feasible at $k=0$. 
Then, Problem~\eqref{eq:SNMPC_shrinking} is feasible and the chance constraints~\eqref{eq:prob_state} are satisfied for the resulting closed-loop system for all $k\in\mathbb{I}_{[0,N-1]}$. 
Furthermore, the closed-loop cost satisfies the following bound:
\begin{align}
\label{eq:SMPC_performance_shrinking}
&\mathbb{E}_{w(0:N-1)}\left[\mathcal{J}_N(x(0:N),u(0:N-1))\right]\nonumber\\
\leq& \mathbb{E}_{w(0:N-1)}\left[\mathcal{J}_N(\mathbf{x}_{0:N|0}^\star,\mathbf{u}_{0:N-1|0}^\star)\right]. 
\end{align}
\end{theorem}
\begin{proof}
\textbf{Recursive feasibility:} Given an optimal input sequence $\mathbf{u}^\star_{0:N-1|k}$ at some time $k\in\mathbb{I}_{[0,N-2]}$, $\mathbf{u}_{0:N-1|k+1}=\mathbf{u}^\star_{0:N-1|k}$ is a feasible solution to Problem~\eqref{eq:SNMPC_shrinking} at time $k+1$. 
In particular, the constraints~\eqref{eq:SNMPC_shrinking_state} remain unaltered, the added constraint in~\eqref{eq:SNMPC_shrinking_input} remains valid with $\mathbf{u}^\star_{k|k}=u(k)$, and 
the stochastic prediction $\mathbf{x}_{\cdot|k+1}$~\eqref{eq:SNMPC_shrinking_prediction} does not affect feasibility.\\
\textbf{Closed-loop chance constraint satisfaction:} First note that Problem~\eqref{eq:SNMPC_shrinking} yields a causal policy $u(k)=\pi_k(w(0:k-1))$, considering also~\eqref{eq:SNMPC_shrinking_input} and independently distributed noise $w$ (Asm.~\ref{asm:disturbance}). 
For any $k\in\mathbb{I}_{[0,N-1]}$, feasibility of Problem~\eqref{eq:SNMPC_shrinking} ensures $\mathcal{R}_k(x_0,u(0:k-1)=\mathcal{R}_k(x_0,\mathbf{u}^\star_{0:k-1|k})\subseteq\mathbb{X}$.
Assumption~\ref{asm:PRS} with the PRS definition (Def.~\ref{def:PRS}) yields 
\begin{align*}
\prob{x(k)\in\mathbb{X}}\stackrel{\eqref{eq:SNMPC_shrinking_state}}{\geq}& \prob{x(k)\in\mathcal{R}_k(x_0,\mathbf{u}_{0:k-1|k}^\star)}\\
\stackrel{\eqref{eq:SNMPC_shrinking_input}}{=} &\prob{x(k)\in\mathcal{R}_k(x_0,{u}(0:k-1))}\stackrel{\eqref{eq:PRS}}{\geq }p, 
\end{align*} 
i.e., the chance constraint~\eqref{eq:prob_state} holds for the closed-loop system. \\
\textbf{Expected cost:} The following derivation follows standard SMPC arguments, cf.~\cite[Thm.~3]{Hewing2020indirect} and \cite[Prop.~11]{mcallister2022nonlinear}. 
Denote $\mathbf{x}^\star_{0:k|k}=x(0:k)$ and by $\mathbf{x}^\star_{k+1:N|k}$ the stochastic prediction~\eqref{eq:SNMPC_shrinking_init}--\eqref{eq:SNMPC_shrinking_prediction_w} using $\mathbf{u}^\star_{k:k+N-1|k}$. 
Consider an arbitrary time $k\in\mathbb{I}_{[0,N-1]}$ and note that Problem~\eqref{eq:SNMPC_shrinking} minimizes the cost
\begin{align*}
&\mathbb{E}_{w(k:N-1)}\left[\mathcal{J}_{N}(\mathbf{x}_{\cdot|k},\mathbf{u}_{\cdot|k})|w(0:k-1)\right],
\end{align*}
where the first $k$ elements of the cost are constant and the conditioning on $w(0:k-1)$ uniquely invokes $x(0:k)$ and $u(0:k-1)$. 
For any $k\in\mathbb{I}_{[0,N-2]}$, it holds that 
\begin{align}
\label{eq:expected_cost_shrinking_step}
&\mathbb{E}_{w(k:N-1)}\left[\mathcal{J}_N(\mathbf{x}^\star_{\cdot|k},\mathbf{u}_{\cdot|k}^\star)\mid w(0:k-1)\right]\nonumber\\
=&\mathbb{E}_{w(k)}\left[\mathbb{E}_{w(k+1:N-1)}\left[\mathcal{J}_N(\mathbf{x}_{\cdot|k}^\star,\mathbf{u}_{\cdot|k}^\star)\mid w(0:k)\right]\right]\nonumber\\
\geq &\mathbb{E}_{w(k)}\left[\mathbb{E}_{w(k+1:N-1)}[\mathcal{J}_N(\mathbf{x}^\star_{\cdot|k+1},\mathbf{u}^\star_{\cdot|k+1})\mid w(0:k)]\right],
\end{align}
where the equality uses the law of iterated expectation. 
The inequality uses the fact that the optimal solution to Problem~\eqref{eq:SNMPC_shrinking} at time $k$ is a feasible solution at time $k+1$ and the inner expectation corresponds to the objective of Problem~\eqref{eq:SNMPC_shrinking} at time $k+1$. 
Iteratively applying Inequality~\eqref{eq:expected_cost_shrinking_step} for $k\in\mathbb{I}_{[0,N-2]}$ yields 
\begin{align*}
&\mathbb{E}_{w(0:N-1)}[\mathcal{J}_N(x(0:N),u(0:N-1))]\\
\stackrel{\eqref{eq:SNMPC_shrinking_input}}{=}&\mathbb{E}_{w(0:N-2)}[\mathbb{E}_{w(N-1)}[\mathcal{J}_N(\mathbf{x}^\star_{\cdot|N-1},\mathbf{u}^\star_{\cdot|N-1})\mid w(0:N-2)]]\\
\leq &\dots\leq\mathbb{E}_{w(0)}[\mathbb{E}_{w(1:N-1)}[\mathcal{J}_N(\mathbf{x}^\star_{\cdot|1},\mathbf{u}^\star_{\cdot|1})\mid w(0)]]\\
\leq & \mathbb{E}_{w(0:N-1)}[\mathcal{J}_N(\mathbf{x}^\star_{\cdot|0},\mathbf{u}^\star_{\cdot|0})].\qedhere
\end{align*}
\end{proof}
Inequality~\eqref{eq:SMPC_performance_shrinking} ensures that the (expected) closed-loop cost is no larger than the cost of the open-loop optimal input sequence, i.e., re-optimization improves performance. 
Furthermore, Theorem~\ref{thm:SMPC_shrinking} ensures that the constrained optimization problem~\eqref{eq:SNMPC_shrinking} is feasible and the chance constraints~\eqref{eq:prob_state} are satisfied for all $k\in\mathbb{I}_{[0,N-1]}$. 
Thus, this approach achieves all the desired properties,~\ref{enum_desired_feas},~\ref{enum_desired_chance},~\ref{enum_desired_cost}. 
However, the explicit dependence of Problem~\eqref{eq:SNMPC_shrinking} on all past inputs limits applicability to short horizon problems. 
Furthermore, Problem~\eqref{eq:SNMPC_shrinking} requires the evaluation of the expected cost~\eqref{eq:SNMPC_shrinking_cost} over the stochastic prediction~\eqref{eq:SNMPC_shrinking_init}--\eqref{eq:SNMPC_shrinking_prediction_w}. 
Practical implementations typically use sampling-based approximations, i.e., replacing the expectation by the empirical mean computed through multiple samples~\cite{kantas2009sequential,sehr2017particle}, which drastically increases the computational complexity~\cite{mesbah2016stochastic}. 
The next section addresses these issues by providing a tractable formulation.

%!TEX root = main.tex
\section{Tractable SMPC scheme}  
\label{sec:SNMPC_tractable}
In this section, we derive a tractable receding-horizon SMPC by introducing a nominal state $z$ and over-approximating the PRS $\mathcal{R}_k$ with a simpler parametrization. 
We define the nominal dynamics 
\begin{align}
\label{eq:nominal}
z(k+1)=f(z(k),u(k),0), \quad z(0)=x_0,
\end{align}
which correspond to~\eqref{eq:sys} with $w(\cdot)\equiv 0$. 
The following assumption ensures that we can over-approximate the PRS $\mathcal{R}_k$ by a sequence of sets centred around the nominal state $z$. 
\begin{assumption}
\label{asm:PRS_nom}
(PRS over-approximation) 
There exists a sequence of known compact sets  $\mathbb{D}_k\subseteq\mathbb{R}^n$, $k\in\mathbb{I}_{\geq 0}$, such that
for any $x_0\in\mathbb{R}^n$, $k\in\mathbb{I}_{\geq 0}$, and any ${u}(0:k-1)\in\mathbb{U}^k$:
\begin{align}
\label{eq:PRS_nom}
\mathcal{R}_k(x_0,{u}(0:k-1))\subseteq \{z(k)\}\oplus \mathbb{D}_k,
\end{align}
with $z(k)$ according to~\eqref{eq:nominal} and $\mathcal{R}_k$ from Assumption~\ref{asm:PRS}. 
\end{assumption}
The PRS construction in Section~\ref{sec:CCM} provides a constructive choice for the sets $\mathbb{D}_k$ satisfying Assumption~\ref{asm:PRS_nom}. 
With this more structured over-approximation, we can reformulate the constraint~\eqref{eq:SNMPC_shrinking_state} using the nominal state $z$ and a sequence of tightened constraints $\bar{\mathbb{X}}_i$ satisfying $\mathbb{D}_i\oplus\bar{\mathbb{X}}_i\subseteq \mathbb{X}$, $i\in\mathbb{I}_{\geq 0}$, see Section~\ref{sec:algorithm_tightening} for an efficient design. 
Given these ingredients, we can formulate the proposed tractable finite-horizon SMPC scheme. 
At each time $k\in\mathbb{I}_{\geq 0}$, we solve the following optimization problem using the measured state $x(k)$ and the nominal state $z(k)$:
\begin{subequations}
\label{eq:SNMPC}
\begin{align}
\label{eq:SNMPC_cost}
\min_{\mathbf{u}_{k:k+N-1|k}\in\mathbb{U}^N}&~\mathcal{J}_{N}(\mathbf{x}_{k:k+N|k},\mathbf{u}_{k:k+N-1|k})\\
\label{eq:SNMPC_init_z}
\text{s.t. }&\mathbf{z}_{k|k}=z(k),\\
\label{eq:SNMPC_dynamics_nominal}
&\mathbf{z}_{i+1|k}=f(\mathbf{z}_{i|k},\mathbf{u}_{i|k},0),\\
\label{eq:SNMPC_init_x}
&\mathbf{x}_{k|k}=x(k),\\
\label{eq:SNMPC_dynamics_mean}
&\mathbf{x}_{i+1|k}=f(\mathbf{x}_{i|k},\mathbf{u}_{i|k},0),\\
\label{eq:SNMPC_tightened}
&\mathbf{z}_{i|k}\in\bar{\mathbb{X}}_{i},\\
\label{eq:SNMPC_terminal}
&\mathbf{z}_{k+N|k}\in\mathbb{X}_{\mathrm{f}},\\
&i\in\mathbb{I}_{[k,k+N-1]}.
\end{align}
\end{subequations}
This problem uses a finite (receding) horizon $N\in\mathbb{I}_{\geq 1}$ and a later specified closed terminal set $\mathbb{X}_{\mathrm{f}}\subseteq\mathbb{X}$.
The closed-loop system is given by applying the minimizer\footnote{%
For any $x(k),z(k)\in\mathbb{R}^n$, a minimizer exists since the cost is a continuous function of $\mathbf{u}_{k:k+N-1|k}\in\mathbb{U}^N$ and $\mathbb{U}$ is compact. 
} $u(k)=\mathbf{u}^\star_{k|k}$ to system~\eqref{eq:sys} and the nominal dynamics~\eqref{eq:nominal}.  
The tightened constraints~\eqref{eq:SNMPC_tightened} and the terminal constraint~\eqref{eq:SNMPC_terminal} are posed on the nominal prediction $\mathbf{z}$~\eqref{eq:SNMPC_dynamics_nominal}, which is initialized in~\eqref{eq:SNMPC_init_z} independent of the new measured state $x(k)$, analogous to the constraints in Problem~\eqref{eq:SNMPC_shrinking}.  
The measured state $x(k)$ is used in a separate certainty-equivalent state prediction~\eqref{eq:SNMPC_init_x}--\eqref{eq:SNMPC_dynamics_mean}, which determines the cost~\eqref{eq:SNMPC_cost}.  
The terminal set $\mathbb{X}_{\mathrm{f}}\subseteq\mathbb{X}$ and terminal cost $V_{\mathrm{f}}$ need to be chosen appropriately according to the following conditions. 
\begin{assumption}(Terminal set and cost)
\label{asm:terminal}
There exists an input $u_{\mathrm{f}}\in\mathbb{U}$, such that: 
\begin{enumerate}[label=\alph*)]
\item (Positive invariance) $f(x,u_{\mathrm{f}},0)\in\mathbb{X}_{\mathrm{f}}$ $\forall x\in\mathbb{X}_{\mathrm{f}}{;}$
\label{term_invar}
\item (Constraint satisfaction) $\mathbb{X}_{\mathrm{f}}\subseteq\bar{\mathbb{X}}_{k}$, $k\in\mathbb{I}_{\geq 0}{;}$ 
\label{term_con}
\item (Lyapunov) 
$V_{\mathrm{f}}(f(x,u_{\mathrm{f}},0))\leq V_{\mathrm{f}}(x)-\ell(x,u_{\mathrm{f}})$ $\forall x\in\mathbb{R}^n$. 
\label{term_CLF}
\end{enumerate}
\end{assumption}
The global condition in~\ref{term_CLF} is in accordance with existing results for stochastic MPC for unbounded noise~\cite{chatterjee2014stability,Hewing2020indirect}, see also the discussion on necessity in~\cite{mayne2019stabilizing}. 
A constructive design satisfying Assumption~\ref{asm:terminal} is provided in Section~\ref{sec_term}. 

The following theorem derives the closed-loop properties of the proposed SMPC formulation. 
\begin{theorem}
\label{thm:SMPC}
Let Assumptions~\ref{asm:disturbance}, \ref{asm:setup_regularity}\ref{asm:setup_regularity_dynamics}--\ref{asm:setup_regularity_origin}, \ref{asm:PRS}, \ref{asm:PRS_nom}, and \ref{asm:terminal} hold. 
Suppose that Problem~\eqref{eq:SNMPC} is feasible at $k=0$. 
Then, Problem~\eqref{eq:SNMPC} is feasible and the chance constraints~\eqref{eq:prob_state} are satisfied for the resulting closed-loop system for all $k\in\mathbb{I}_{\geq 0}$. 
Furthermore, there exists a function $\sigma\in\mathcal{K}_\infty$, such that the closed-loop cost satisfies the following bound:
\begin{align}
\label{eq:SNMPC_performance_average}
\limsup_{T\rightarrow\infty}\mathbb{E}_{w(0:T-1)}\left[\dfrac{1}{T}\sum_{k=0}^{T-1}\ell(x(k),u(k))\right]\leq \sigma(\tr{\Sigma_{\mathrm{w}}}).
\end{align}
\end{theorem}
\begin{proof}
\textbf{Recursive feasibility:} Given an optimal input sequence $\mathbf{u}^\star_{k:k+N-1|k}$ at some time $k\in\mathbb{I}_{\geq 0}$, denote $\mathbf{u}^\star_{k+N|k}=u_{\mathrm{f}}\in\mathbb{U}$ and consider the candidate input sequence 
$\mathbf{u}_{i|k+1}=\mathbf{u}^\star_{i|k}$, $i\in\mathbb{I}_{[k+1,k+N]}$. 
The corresponding nominal state sequence is given by $\mathbf{z}_{i|k+1}=\mathbf{z}_{i|k}^\star$, $i\in\mathbb{I}_{[k+1,k+N]}$ using $u(k)=\mathbf{u}_{k|k}^\star$, \eqref{eq:SNMPC_init_z}, \eqref{eq:SNMPC_dynamics_nominal}, and \eqref{eq:nominal}. 
The noise $w(k)$ only affects the state $x(k+1)$ and hence $\mathbf{x}_{\cdot|k+1}$, which is not subject to any constraints. 
Feasibility of the candidate solution and thus recursive feasibility follows with standard nominal MPC arguments~\cite{rawlings2017model} using Assumptions~\ref{asm:terminal} \ref{term_invar}--\ref{term_con}, 
$\mathbf{z}_{k+N|k+1}=\mathbf{z}_{k+N|k}^\star\in\mathbb{X}_{\mathrm{f}}\subseteq\bar{\mathbb{X}}_{k+N}$, and 
$\mathbf{z}_{k+N+1|k+1}=f(\mathbf{z}_{k+N|k}^\star,u_{\mathrm{f}},0)\in\mathbb{X}_{\mathrm{f}}$.\\
\textbf{Closed-loop chance constraint satisfaction:}
Consider the chance constraint~\eqref{eq:prob_state} for some $k\in\mathbb{I}_{\geq 0}$ and $z(k)=\mathbf{z}_{k|k}^\star\in\bar{\mathbb{X}}_k$ from Problem~\eqref{eq:SNMPC}. 
Assumption~\ref{asm:PRS_nom} ensures
\begin{align*}
\mathcal{R}_k(x_0,u(0:k-1))\stackrel{\eqref{eq:PRS_nom}}{\subseteq} \{z(k)\}\oplus\mathbb{D}_k\stackrel{\eqref{eq:SNMPC_init_z},\eqref{eq:SNMPC_tightened}}{\subseteq} \bar{\mathbb{X}}_k\oplus\mathbb{D}_k\subseteq \mathbb{X}.
\end{align*}
The inputs $u(0:k-1)$ generated by the SMPC~\eqref{eq:SNMPC} correspond to a causal policy and hence Assumption~\ref{asm:PRS} yields
\begin{align*}
\prob{x(k)\in\mathbb{X}}\geq \prob{x(k)\in\mathcal{R}_k(x_0,u(0:k-1))}\stackrel{\eqref{eq:PRS}}{\geq} p,
\end{align*}
i.e., the chance constraint~\eqref{eq:prob_state} holds. \\
\textbf{Cost bound:} This proof consistent of three steps: \textbf{(i)} deriving a bound on the one step decrease using the candidate solution; \textbf{(ii)} bounding the optimal cost of Problem~\eqref{eq:SNMPC} for any $x(k)\in\mathbb{R}^n$; \textbf{(iii)} deriving the asymptotic expected cost bound~\eqref{eq:SNMPC_performance_average}.\\  
\textbf{(i):} 
We denote the state sequence satisfying~\eqref{eq:SNMPC_init_x}--\eqref{eq:SNMPC_dynamics_mean} with the candidate input $\mathbf{u}_{\cdot|k+1}$ by $\mathbf{x}_{\cdot|k+1}$ and define $\mathbf{x}_{k+N+1|k}^\star=f(\mathbf{x}_{k+N|k}^\star,\mathbf{u}^\star_{k+N|k},0)$. 
Lipschitz continuity (Asm.~\ref{asm:setup_regularity}) of the dynamics~\eqref{eq:SNMPC_init_x}--\eqref{eq:SNMPC_dynamics_mean} and \eqref{eq:sys} implies
\begin{align*}
&\|\mathbf{x}_{k+1|k+1}-\mathbf{x}_{k+1|k}^\star\|\leq L_{\mathrm{f}}\|w(k)\|,\\
&\|\mathbf{x}_{i|k+1}-\mathbf{x}_{i|k}^\star\|\leq L_{\mathrm{f}}^{i-k}\|w(k)\|,~i\in\mathbb{I}_{[k+1,k+N+1]},
\end{align*}
with Lipschitz constant $L_{\mathrm{f}}\geq 0$. 
For any $x,y\in\mathbb{R}^n$, any positive semi-definite matrix $S\in\mathbb{R}^{n\times n}$, and any $\epsilon>0$: 
\begin{align}
\label{eq:quad_epsilon}
\dfrac{1}{1+\epsilon}\|x+y\|_S^2\leq \|x\|_S^2+\dfrac{1}{\epsilon}\|y\|_S^2.
\end{align}
Hence, for any $\epsilon>0$, the quadratic cost (Asm.~\ref{asm:setup_regularity}) satisfies
\begin{align*}
&\dfrac{1}{1+\epsilon}\ell(\mathbf{x}_{i|k+1},\mathbf{u}_{i|k+1})\leq \ell(\mathbf{x}_{i|k}^\star,\mathbf{u}_{i|k}^\star)+\frac{1}{\epsilon}\bar{Q} L_{\mathrm{f}}^{i-k}\|w(k)\|^2,\\
&\dfrac{1}{1+\epsilon}V_{\mathrm{f}}(\mathbf{x}_{k+N+1|k+1})\leq V_{\mathrm{f}}(\mathbf{x}_{k+N+1|k}^\star)+\frac{1}{\epsilon}\bar{P}L_{\mathrm{f}}^{N+1}\|w(k)\|^2,
\end{align*}
where $\bar{P}\geq\bar{Q}>0$ are the maximal eigenvalues of $P,Q$, respectively. 
Let us denote the optimal cost of Problem~\eqref{eq:SNMPC} at time $k$ by $\mathcal{J}_N^\star(k)=\mathcal{J}_N(\mathbf{x}_{k:k+N|k}^\star,\mathbf{u}_{k:k+N-1|k}^\star)$. 
The feasible candidate solution implies
\begin{allowdisplaybreaks}
\begin{align}
\label{eq:cost_decrease_one_step_intermediate}
&\dfrac{1}{1+\epsilon}\mathcal{J}_N^\star(k+1)\nonumber\\
\leq &\dfrac{1}{1+\epsilon}\mathcal{J}_N(\mathbf{x}_{k+1:k+N+1|k+1},\mathbf{u}_{k+1:k+N|k+1})\nonumber\\
\leq&\sum_{i=k+1}^{k+N}\ell(\mathbf{x}_{i|k}^\star,\mathbf{u}_{i|k}^\star)+V_{\mathrm{f}}(\mathbf{x}_{k+N+1|k}^\star)\nonumber\\
&+\dfrac{1}{\epsilon}\|w(k)\|^2\left(\sum_{i=k+1}^{k+N} \bar{Q}L_{\mathrm{f}}^{i-k} +\bar{P} L_{\mathrm{f}}^{N+1}\right)\nonumber\\
\stackrel{\mathrm{Asm.~\ref{asm:terminal}\ref{term_CLF}}}{\leq} &\sum_{i=k+1}^{k+N-1}\ell(\mathbf{x}_{i|k}^\star,\mathbf{u}_{i|k}^\star)+V_{\mathrm{f}}(\mathbf{x}_{k+N|k}^\star)\nonumber\\
&+\dfrac{1}{\epsilon}\|w(k)\|^2\underbrace{\left(\sum_{i=1}^{N} \bar{Q} L_{\mathrm{f}}^{i} +\bar{P}L_{\mathrm{f}}^{N+1}\right)}_{=:c_{\mathcal{J}}}\nonumber\\
=&\mathcal{J}_N^\star(k)-\ell(x(k),u(k))+\dfrac{c_{\mathcal{J}}}{\epsilon}\|w(k)\|^2,
\end{align}
\end{allowdisplaybreaks}
with $x(k)=\mathbf{x}_{k|k}^\star$, $u(k)=\mathbf{u}_{k|k}^\star$. \\
\textbf{(ii):} Next, we derive an upper bound on $\mathcal{J}_N^\star(k)$ for any $x(k)\in\mathbb{R}^n$, using feasibility of Problem~\eqref{eq:SNMPC}. 
Notably, the feasibility of Problem~\eqref{eq:SNMPC} does not imply a uniform bound on the state $x(k)$, as constraints are only imposed on the nominal prediction $\mathbf{z}$. 
This is also the main reason why existing techniques, such as~\cite{mcallister2022nonlinear}, cannot be leveraged to derive bounds on $\mathcal{J}_N^\star(k)$. 
We define an auxiliary input sequence $\mathbf{\tilde{u}}_{i|k}=u_{\mathrm{f}}\in\mathbb{U}$, $i\in\mathbb{I}_{[k,k+N-1]}$, with the corresponding state sequence $\mathbf{\tilde{x}}_{i|k}$ according to~\eqref{eq:SNMPC_init_x}--\eqref{eq:SNMPC_dynamics_mean}.
The cost of this sequence satisfies
\begin{align}
\label{eq:cost_candidate_global}
\mathcal{J}_N(\mathbf{\tilde{x}}_{\cdot|k},\mathbf{\tilde{u}}_{\cdot|k})\leq V_{\mathrm{f}}(\mathbf{\tilde{x}}_{0|k})=\|x(k)\|_P^2
\end{align}
using the global properties of the terminal cost (Asm.~\ref{asm:terminal}~\ref{term_CLF}) recursively. 
Although this input sequence is not a feasible solution to Problem~\eqref{eq:SNMPC}, we can leverage it to bound the cost of the optimal solution. 
Using the bound $\bar{u}=\max_{u\in\mathbb{U}}\|u_{\mathrm{f}}-u\|$ and Lipschitz continuity, we have 
\begin{align*}
&\|\mathbf{x}^\star_{k+i|k}-\mathbf{\tilde{x}}_{k+i|k}\|\leq\sum_{j=0}^{i-1} L_{\mathrm{f}}^{i-j}\|\mathbf{u}_{k+j|k}^\star-u_{\mathrm{f}}\|\leq \bar{u}\sum_{j=0}^{i-1} L_{\mathrm{f}}^{i-j}.
\end{align*}
Applying~\eqref{eq:quad_epsilon} with $\epsilon=1$, we can bound the cost of the optimal sequence by 
\begin{align*}
&\ell(\mathbf{x}^\star_{i|k},\mathbf{u}_{i|k})\leq 2\ell(\mathbf{\tilde{x}}_{i|k},\mathbf{\tilde{u}}_{i|k})+2\bar{Q}\bar{u}^2\sum_{j=0}^{i-1} L_{\mathrm{f}}^{2(i-j)},\\
&V_{\mathrm{f}}(\mathbf{x}^\star_{k+N|k})\leq 2V_{\mathrm{f}}(\mathbf{\tilde{x}}_{k+N|k})+2\bar{P}\bar{u}^2\sum_{j=0}^{N-1} L_{\mathrm{f}}^{2(i-j)}.
\end{align*}
Thus, the optimal cost satisfies
\begin{align}
\label{eq:cost_upper_bound}
&\mathcal{J}_N^\star(k)=\mathcal{J}_N(\mathbf{x}^\star_{\cdot|k},\mathbf{u}^\star_{\cdot|k})\nonumber\\
\leq& 2\mathcal{J}_N(\mathbf{\tilde{x}}_{\cdot|k},\mathbf{\tilde{u}}_{\cdot|k})
+\underbrace{2\bar{u}^2\max\{\bar{Q},\bar{P}\}\sum_{i=0}^N\sum_{j=0}^{i-1} L_{\mathrm{f}}^{2(i-j)}}_{=:c_{\mathcal{J},2}}\nonumber\\
\stackrel{\eqref{eq:cost_candidate_global}}{\leq}& 2 \|x(k)\|_P^2+c_{\mathcal{J},2}\stackrel{\mathrm{Asm.}~\ref{asm:setup_regularity}}{\leq} c_{\ell}\ell(x(k),u(k))+c_{\mathcal{J},2},
\end{align}
with $c_\ell:=2\bar{P}/\underline{Q}>1$, where $\underline{Q}>0$ is the smallest eigenvalue of $ Q\succ 0$. \\
\textbf{(iii):} 
Inequalities~\eqref{eq:cost_decrease_one_step_intermediate} and \eqref{eq:cost_upper_bound} yield
\begin{align}
\label{eq:expected_decrease_one_step_almost_final}
&\mathcal{J}_N^\star(k+1)-\mathcal{J}_N^\star(k)\\
\stackrel{\eqref{eq:cost_decrease_one_step_intermediate}}{\leq}&\epsilon\mathcal{J}_N^\star(k)-(1+\epsilon)\ell(x(k),u(k))+\dfrac{1+\epsilon}{\epsilon}c_{\mathcal{J}}\|w(k)\|^2\nonumber\\
\stackrel{\eqref{eq:cost_upper_bound}}{\leq} &-\ell(x(k),u(k))(1-\epsilon c_\ell) + \epsilon c_{\mathcal{J},2} + \dfrac{1+\epsilon}{\epsilon}c_{\mathcal{J}}\|w(k)\|^2.\nonumber
\end{align}
Consider $\epsilon=\min\{\sqrt{\tr{\Sigma_{\mathrm{w}}}},\frac{1}{2c_\ell}\}\in(0,1)$, which satisfies 
$(1-\epsilon c_\ell)\geq \frac{1}{2}$, $\epsilon\leq \sqrt{\tr{\Sigma_{\mathrm{w}}}}$, and 
\begin{align}
\label{eq:proof_bounds_epsilon}
 \dfrac{1+\epsilon}{\epsilon}\tr{\Sigma_{\mathrm{w}}}\leq \dfrac{2}{\epsilon}\tr{\Sigma_{\mathrm{w}}}=2\max\left\{\sqrt{\tr{\Sigma_{\mathrm{w}}}},\frac{\tr{\Sigma_{\mathrm{w}}}}{2 c_\ell}\right\}.
\end{align}
This yields
\begin{align*}
&\mathbb{E}_{w(k)}\left[\mathcal{J}_N^\star(k+1)-\mathcal{J}_N^\star(k) \phantom{\frac{1}{2}}\right.\\%phantom ensures E has correct heights
&\phantom{\mathbb{E}_{w(k)}[}\left.+\frac{1}{2}\ell(x(k),u(k))\mid w(0:k-1)\right]\\%phantom ensures correct alignement
\stackrel{\eqref{eq:expected_decrease_one_step_almost_final}}{\leq} &\mathbb{E}_{w(k)}\left[\epsilon c_{\mathcal{J},2} +\dfrac{1+\epsilon}{\epsilon}c_{\mathcal{J}}\|w(k)\|^2 \mid w(0:k-1)\right]\\
\stackrel{\mathrm{Asm.}~\ref{asm:disturbance}}{\leq} &\epsilon c_{\mathcal{J},2} +\dfrac{1+\epsilon}{\epsilon}c_{\mathcal{J}} \tr{\Sigma_{\mathrm{w}}}\\
\stackrel{\eqref{eq:proof_bounds_epsilon}}{\leq }& c_{\mathcal{J},2}  \sqrt{\tr{\Sigma_{\mathrm{w}}}}+2c_{\mathcal{J}}\max\left\{\sqrt{\tr{\Sigma_{\mathrm{w}}}},\frac{\tr{\Sigma_{\mathrm{w}}}}{2 c_\ell}\right\}\\
=:&\tilde{\sigma}(\tr{\Sigma_{\mathrm{w}}}),\quad \tilde{\sigma}\in\mathcal{K}_\infty.
\end{align*}
Applying the law of iterated expectation for $k\in\mathbb{I}_{[0,T-1]}$ yields
\begin{align*}
0&\leq \mathbb{E}_{w(0:T-1)}[\mathcal{J}_N^\star(T)]\\
\leq &\mathcal{J}_N^\star(0)+T\tilde{\sigma}(\tr{\Sigma_{\mathrm{w}}})-\frac{1}{2}\mathbb{E}_{w(0:T-1)}\left[\sum_{k=0}^{T-1}\ell(x(k),u(k))\right],
\end{align*}
where the first inequality used non-negativity of the cost. 
Given $x(0)=z(0)\in\mathbb{R}^n$, $f$ Lipschitz continuous, $\mathbb{U}$ compact, and $\mathcal{J}_N$ quadratic, we have $\mathcal{J}_N^\star(0)<\infty$.
Thus, dividing by $T$ and taking the limit $T\rightarrow\infty$ yields~\eqref{eq:SNMPC_performance_average} with $\sigma=2\tilde{\sigma}\in\mathcal{K}_\infty$. 
\end{proof}
Overall, Theorem~\ref{thm:SMPC} provides all the desired closed-loop guarantees: recursive feasibility~\ref{enum_desired_feas}, satisfaction of chance constraints~\ref{enum_desired_chance}, and a bound on the expected cost~\ref{enum_desired_cost}. 
Notably, we ensure recursive feasibility regardless of the magnitude of the noise realizations during closed-loop operation. 
The fact that the asymptotic expected cost scales with the variances mirrors existing results in SMPC~\cite{mcallister2022nonlinear,Hewing2020indirect}, see Section~\ref{sec:discussion} for a more detailed discussion.
In the next section, we discuss how to design the PRS (Asm.~\ref{asm:PRS}--\ref{asm:PRS_nom}).

%!TEX root = main.tex
 \section{Probabilistic reachable sets using stochastic contraction metrics}
\label{sec:CCM}
In this section, we show how to construct a PRS satisfying Assumptions~\ref{asm:PRS}-\ref{asm:PRS_nom} using contraction metrics. 
First, we use stochastic contraction metrics to bound the expected prediction error (Sec.~\ref{sec:CCM_1}) and then derive a PRS (Sec.~\ref{sec:CCM_2}). 

\subsection{Stochastic contraction metrics}
\label{sec:CCM_1}
Contraction metrics utilize conditions on the Jacobian of the nonlinear dynamics $f$ to ensure incremental stability properties for the nonlinear system~\cite{tsukamoto2020robust}. 
Given that the dynamics $f$ are continuous differentiable and linear in $w$ (Asm.~\ref{asm:setup_regularity}\ref{asm:setup_regularity_CCM}), 
we denote the Jacobian of the dynamics $f$ by 
\begin{align*}
A(x,u):=\left.\dfrac{\partial f}{\partial x}\right|_{(x,u,w)},
\quad 
G=\left.\dfrac{\partial f}{\partial w}\right|_{(x,u,w)}.
\end{align*}
The following theorem derives bounds on the expected error using stochastic contraction metrics. 
\begin{theorem}
\label{thm:CCM}
Let Assumptions~\ref{asm:disturbance} and \ref{asm:setup_regularity}\ref{asm:setup_regularity_CCM} hold. 
Consider a state-dependent matrix $M:\mathbb{R}^n\rightarrow\mathbb{R}^{n\times n}$, which satisfies
\begin{subequations}
\label{eq:CCM_all_4}
\begin{align}
\label{eq:CCM_anhiling}
\dfrac{\partial}{\partial w} M(x+G\cdot w)=0 \quad \forall (x,w)\in\mathbb{R}^n\times\mathbb{R}^q.
\end{align}
Suppose there exist positive definite matrices $\underline{M},\overline{M}\in\mathbb{R}^{n\times n}$, a contraction rate $\rho\in[0,1)$, and a constant $\bar{w}\geq 0$, such that the following conditions hold for all $x\in\mathbb{R}^n$, $u\in\mathbb{U}$:
\begin{align}
\label{eq:CCM_M_bounds}
 &\underline{M}\preceq M(x)\preceq \overline{M},\\
\label{eq:CCM_rho}
& A(x,u)^\top M(f(x,u,0)) A(x,u)\preceq \rho M(x),\\
\label{eq:CCM_w}
&\tr{\Sigma_{\mathrm{w}}G^\top \bar{M} G}\leq \bar{w}.
\end{align}
\end{subequations}
Then, there exists an incremental Lyapunov function $V_\delta:\mathbb{R}^n\times\mathbb{R}^n\rightarrow\mathbb{R}$, such that 
for any $x,z\in\mathbb{R}^n$, $u\in\mathbb{U}$, $k\in\mathbb{I}_{\geq 0}$: 
\begin{subequations}
\label{eq:increm_bounds}
\begin{align}
\label{eq:increm_bounds_1}
\|x-z\|_{\underline{M}}^2\leq V_\delta(x,z)\leq &\|x-z\|_{\overline{M}}^2,\\
\label{eq:increm_bounds_2}
\mathbb{E}_{w(k)}\left[V_\delta(f(x,u,w(k)),f(z,u,0))\right]\leq &\rho V_\delta(x,z)+\bar{w}.
\end{align}
\end{subequations}
\end{theorem}
\begin{proof}
We first define $V_\delta$ and show the bounds~\eqref{eq:increm_bounds_1} before deriving the decrease condition~\eqref{eq:increm_bounds_2}.\\
\textbf{Part (i): }Let us denote the set of piece-wise smooth curves $\gamma:[0,1]\rightarrow\mathbb{R}^n$ that satisfy $\gamma(0)=z$ and $\gamma(1)=x$ by $\Gamma(z,x)$. 
We define the incremental Lyapunov function $V_\delta$ as the Riemannian energy corresponding to $M(x)$, i.e.,
\begin{align}
\label{eq:increm_V_delta_def}
V_\delta(x,z):=\min_{\gamma\in\Gamma(z,x)}\int_0^1 \left.\dfrac{\partial \gamma}{\partial s}\right|_s^\top M(\gamma(s)) \left.\dfrac{\partial \gamma}{\partial s}\right|_s~\mathrm{d}s.
\end{align}
A minimizer is called a geodesic $\gamma^\star$, which exists due to the uniform lower bound~\eqref{eq:CCM_M_bounds}, see~\cite[Lemma~1]{manchester2017control}. 
Inequalities~\eqref{eq:increm_bounds_1} follow from~\eqref{eq:CCM_M_bounds} using standard arguments, see, e.g.,~\cite[Equ.~(49)--(51)]{schiller2023lyapunov}. \\
\textbf{Part (ii): }
Denote $x^+=f(x,u,w(k))$, $z^+=f(z,u,0)$. 
We derive a bound on $V_\delta(x^+,z^+)$ using the candidate curve $\gamma^+(s)=f(\gamma^\star(s),u,\gamma_w(s))$ with $\gamma_w(s)=s\cdot w(k)$. 
With some abuse of notation, we abbreviate 
$A(s)=A(\gamma^\star(s),u)$, $M(s)=M(\gamma^\star(s))$, 
$M^+(s)=M(\gamma^+(s))$. 
We denote the derivative of the geodesic by $\gamma^\star_{\mathrm{s}}(s)=\left.\dfrac{\partial \gamma^\star}{\partial s}\right|_s$. The derivate of the candidate curve $\gamma^+$ satisfies
\begin{align}
\label{eq:geodesic_candidate_derivative}
\gamma_{\mathrm{s}}^+(s):=\left.\frac{\partial \gamma^+}{\partial s}\right|_s=A(s)\gamma^\star_{\mathrm{s}}(s)+G w(k).
\end{align} 
Note that $\gamma^\star\in\Gamma(z,x)$ implies that $\gamma^+\in\Gamma(z^+,x^+)$ with
$\gamma^+(1)=x^+$, $\gamma^+(0)=z^+$. Hence, the energy of the curve $\gamma^+$ provides an upper bound to the Riemannian energy $V_\delta(x^+,z^+)$ according to~\eqref{eq:increm_V_delta_def}. 
Thus, we have: 
\begin{align*}
& \mathbb{E}_{w(k)}\left[V_\delta(x^+,z^+)\right]
\leq \mathbb{E}_{w(k)}\left[\int_0^1 \gamma^+_{\mathrm{s}}(s)^\top M^+(s) \gamma^+_{\mathrm{s}}(s) \mathrm{d}s\right]\\
\stackrel{\eqref{eq:geodesic_candidate_derivative}}{=}& \mathbb{E}_{w(k)}\left[\int_0^1 \left\|A(s)\gamma^\star_{\mathrm{s}}+G w(k)\right\|_{M^+(s)}^2 \mathrm{d}s\right]\\
=& \mathbb{E}_{w(k)}\left[\int_0^1 
\left\|A(s)\gamma^\star_{\mathrm{s}}(s)\right\|_{M^+(s)}^2 
+\left\|G w(k)\right\|_{M^+(s)}^2 \right.\\
&\left.\phantom{\mathbb{E}_{w(k)}[\int_0^1} +2(A(s)\gamma^\star_{\mathrm{s}}(s))^\top M^+(s) G  w(k) 
\mathrm{d}s\right].
\end{align*}
Next, we bound each of the three terms individually. 
Condition~\eqref{eq:CCM_anhiling} ensures that $M^+(s)$ is independent of $w(k)$, which will be crucial when taking the expectation later. 
For the first term, it holds that
\begin{align*}
&\mathbb{E}_{w(k)}\left[\int_0^1\|A(s)\gamma^\star_{\mathrm{s}}(s)\|_{M^+(s)}^2 \mathrm{d}s\right]\\
\stackrel{\eqref{eq:CCM_rho}}{\leq}& \mathbb{E}_{w(k)}\left[\int_0^1 \rho \|\gamma^\star(s)\|_{M(s)}^2\mathrm{d}s\right]\stackrel{\eqref{eq:increm_V_delta_def}}{=}\rho V_\delta(x,z). \nonumber 
\end{align*} 
For the second term, we get 
\begin{align*}
&\mathbb{E}_{w(k)}\left[\int_0^1\|Gw(k)\|_{M^+(s)}^2 \mathrm{d}s\right]\\
\stackrel{\eqref{eq:CCM_M_bounds}}{\leq}& \mathbb{E}_{w(k)}\left[\int_0^1\| G w(k)\|_{\bar{M}}^2 \mathrm{d}s\right]\\
=&\mathbb{E}_{w(k)}\left[\int_0^1\tr{w(k)w(k)^\top G^\top \bar{M} G} \mathrm{d}s\right]\nonumber\\
=&\int_0^1\tr{\mathbb{E}_{w(k)}\left[w(k)w(k)^\top\right] G^\top \bar{M} G} \mathrm{d}s\nonumber\\
\stackrel{\text{Asm.~}\ref{asm:disturbance}}{\leq}& \int_0^1 \tr{\Sigma_{\mathrm{w}} G^\top \bar{M}G}\mathrm{d}s\stackrel{\eqref{eq:CCM_w}}{\leq}\int_0^1 \bar{w}\mathrm{d}s=\bar{w}.\nonumber
\end{align*}
where the first equality used the cyclic property of the trace and the second equality used linearity of the trace operator. 
For the last term, we get
\begin{align*}
&\mathbb{E}_{w(k)}\left[\int_0^1 2 (A(s)\gamma^\star_{\mathrm{s}}(s))^\top M^+(s) G w(k) \right]~\mathrm{d}s\\
\stackrel{\eqref{eq:CCM_anhiling}}{ =}&\int_0^1 2 (A(s)\gamma^\star_{\mathrm{s}}(s))^\top M^+(s) G \mathbb{E}_{w(k)}[w(k)]~\mathrm{d} s \stackrel{\mathrm{Asm.}~\ref{asm:disturbance}}{=} 0,
\end{align*}
where the first equality leverages the fact that $M^+(s)$ is independent of $w(k)$ due to~\eqref{eq:CCM_anhiling}. 
In combination, these three bounds yield~\eqref{eq:increm_bounds_2}.
\end{proof}
By imposing uniform contraction conditions~\eqref{eq:CCM_all_4} on the Jacobian, Theorem~\ref{thm:CCM} derives bounds on the nonlinear stochastic dynamics~\eqref{eq:sys} without linearization errors, as standard for contraction metrics~\cite{tsukamoto2020robust}. 
Condition~\eqref{eq:CCM_anhiling} restricts the parametrization of $M(x)$. 
Given any matrix $\bar{G}$ that is orthogonal to $G$, Condition~\eqref{eq:CCM_anhiling} is satisfied by choosing $M(x)=\tilde{M}(\bar{G}x)$ with an arbitrarily parametrized matrix $\tilde{M}$. 
In Appendix~\ref{app_CCM}, we provide additional details on the computation of a suitable contraction metric $M$ satisfying~\eqref{eq:CCM_all_4} using linear matrix inequalities (LMIs). 
The following corollary highlights how the result simplifies in the special case of a constant contraction metric $M$. 
\begin{corollary}
\label{cor:CCM_stoch}
Let Assumptions~\ref{asm:disturbance} and \ref{asm:setup_regularity}\ref{asm:setup_regularity_CCM} hold. 
Suppose there exists a positive definite matrix $M\in\mathbb{R}^{n\times n}$, a contraction rate $\rho\in[0,1)$, and a constant $\bar{w}\geq 0$, such that the following conditions hold for all $x\in\mathbb{R}^n$, $u\in\mathbb{U}$:
\begin{subequations}
\label{eq:CCM_all_4_simplified}
\begin{align}
\label{eq:CCM_rho_simplified}
A(x,u)^\top M A(x,u)\preceq& \rho M,\\
\label{eq:CCM_w_simplified}
\tr{\Sigma_{\mathrm{w}}G^\top M G}\leq &\bar{w}.
\end{align}
\end{subequations}
Then, Conditions~\eqref{eq:increm_bounds} hold with $V_\delta(x,z)=\|x-z\|_M^2$. 
\end{corollary}
\begin{proof}
Conditions~\eqref{eq:CCM_anhiling} and~\eqref{eq:CCM_M_bounds} hold trivially with $\bar{M}=\underline{M}=M$. 
Furthermore, Conditions~\eqref{eq:CCM_all_4_simplified} are equivalent to~\eqref{eq:CCM_rho}--\eqref{eq:CCM_w} for $M$ constant.
Lastly, the constant metric $M$ ensures that the geodesic (shortest path) is a straight line, i.e., $\gamma^\star(s)=z+s(x-z)$. 
Hence, the Riemannian energy~\eqref{eq:increm_V_delta_def} reduces to the weighted norm $\|x-z\|_M^2$. 
\end{proof}
\begin{remark}
\label{rk:CCM_stochastic} (Related results on contraction metrics)
For nonlinear dynamics with bounded model-mismatch, robust contraction metrics can be leveraged to derive robust reachable sets around nominal trajectories $z$~\cite{zhao2022tube,singh2023robust,sasfi2023robust}. 
Compared to Theorem~\ref{thm:CCM}, these conditions leverage worst-case bounds instead of the variance bound $\Sigma_{\mathrm{w}}$ (cf. \eqref{eq:CCM_w} vs. \cite[Equ.~(A.7)]{sasfi2023robust}). 
The key difference in the result is the establishment of a deterministic bound in case of bounded model-mismatch compared to the developed expected bound for unbounded stochastic process noise. \\
Considering results for \textit{stochastic} contraction metrics: 
For continuous-time systems and state dependent metrics $M(x)$, bounds on the expected error are derived in~\cite{dani2014observer} and \cite[Thm.~1]{tsukamoto2020robust} using bounds on the derivative of $M(x)$. 
However, these continuous-time results cannot be directly applied to the considered stochastic discrete-time system~\eqref{eq:sys}, where condition~\eqref{eq:CCM_rho} is highly nonlinear in $w$. 
For discrete-time stochastic systems, \cite[Thm.~2]{tsukamoto2020robust} derives a valid expected bound of the form~\eqref{eq:increm_bounds} without requiring Condition~\eqref{eq:CCM_anhiling}.
However, the result can only ensure contraction of the nonlinear system if the condition number of $M$ is sufficiently close to one, thus limiting practical applicability. 
Discrete-time stochastic contraction metrics are also studied in the preprint~\cite{pham2013stochastic}. 
However, the results do not apply to state-dependent metrics $M(x)$ due to the induced correlation with $M(f(x,u,w))$, which we resolved though Condition~\eqref{eq:CCM_anhiling}. 
Overall, results for stochastic contraction metrics in discrete-time comparable to Theorem~\ref{thm:CCM} have been lacking in the literature, see also the overview paper~\cite{tsukamoto2021contraction}. 
Notably, results for stochastic contraction metrics typically rely on linearity in the noise $w$ and Appendix~\ref{app:affine_w} shows how to remove this requirement using a simple lifting.
\end{remark}

\subsection{Probabilist reachable set}
\label{sec:CCM_2}
In the following, we derive a PRS (Def.~\ref{def:PRS}) using the expected bounds of the stochastic contraction metrics (Thm.~\ref{thm:CCM}). 
Due to the nonlinear dynamics, the error $x(k)-z(k)$ does not follow any know distribution and is not even zero-mean. 
Hence, many standard inequalities from the literature on linear SMPC, such as the Chebyshev inequality~\cite{Farina2016soverview,Hewing2020indirect}, cannot be applied. 
The following theorem provides a valid PRS by combing the stochastic contraction metrics (Thm.~\ref{thm:CCM}) with the Markov inequality. 
\begin{theorem}
\label{thm:PRS}
Suppose the conditions in Theorem~\ref{thm:CCM} hold. 
Then, Assumptions~\ref{asm:PRS}--\ref{asm:PRS_nom} hold with
\begin{align}
\label{eq:PRS_design}
&\mathbb{D}_k=\left\{x|~\|x\|_{\underline{M}}^2\leq\sigma_{\mathrm{x},k}\right\},\quad 
\sigma_{\mathrm{x},k}:=\dfrac{1-\rho^k}{1-\rho}\dfrac{\bar{w}}{1-p}\nonumber\\
&\mathcal{R}_k(x_0,u(0:k-1))=z(k)\oplus\mathbb{D}_k,~k\in\mathbb{I}_{\geq 0}, 
\end{align}
with $z(k)$ according to~\eqref{eq:nominal}.
\end{theorem}
\begin{proof}
Given Definition~\ref{def:PRS}, we consider an arbitrary causal policy $u(k)=\pi_k(w(0:k-1))\in\mathbb{U}$, $k\in\mathbb{I}_{\geq 0}$ with the stochastic dynamics~\eqref{eq:sys} and the nominal dynamics~\eqref{eq:nominal}. 
Define $\delta(k):=V_\delta(x(k),z(k))$, $k\in\mathbb{I}_{\geq 0}$. 
Given the fixed initial condition $x(0)=x_0$, the noise $w(0:k-1)$ uniquely invoke $x(0:k)$, $z(0:k)$, $u(0:k)$, $\delta(0:k)$ leveraging causality of the policies $\pi_k$ and the dynamics~\eqref{eq:sys}, \eqref{eq:nominal}. 
For any $k\in\mathbb{I}_{\geq 0}$, the law of iterated expectation ensures
\begin{align}
&\mathbb{E}_{w(0:k)}[\delta(k+1)]\nonumber\\
=&\mathbb{E}_{w(0:k-1)}[\mathbb{E}_{w(k)}[\delta(k+1)\mid w(0:k-1)]]\nonumber\\
=&\mathbb{E}_{w(0:k-1)}[\mathbb{E}_{w(k)}[\delta(k+1)\mid x(k),z(k))]]\nonumber\\
\stackrel{\eqref{eq:increm_bounds_2}}{\leq}& \mathbb{E}_{w(0:k-1)}[\rho \delta(k)+\bar{w}]\nonumber\\
\leq& \dots \leq \rho^{k+1} \delta(0)+\sum_{i=0}^{k}\rho^i\bar{w}=\dfrac{1-\rho^{k+1}}{1-\rho}\bar{w},
\label{eq:CCM_expected_iterated}
\end{align}
where the last equality uses the geometric series and $\delta(0)=0$ due to the fixed initialization in~\eqref{eq:sys}/\eqref{eq:nominal}.
The Markov inequality~\cite{markov1884certain} with $\delta(k)$ non-negative ensures that 
\begin{align}
\label{eq:Markov}
&\prob{\delta(k)\leq \frac{\mathbb{E}_{w(0:k-1)}[\delta(k)]}{1-p}}\geq p,\quad k\in\mathbb{I}_{\geq 0}.
\end{align} 
Combining~\eqref{eq:CCM_expected_iterated} and~\eqref{eq:Markov} yields
\begin{align*}
\prob{\delta(k)\leq \frac{1-\rho^k}{1-\rho}\frac{\bar{w}}{1-p}}\geq p.
\end{align*}
Furthermore, the lower bound~\eqref{eq:increm_bounds_1} yields
\begin{align*}
\prob{\|x(k)-z(k)\|_{\underline{M}}^2\leq \frac{1-\rho^k}{1-\rho}\frac{\bar{w}}{1-p}}\geq p.
\end{align*}
Thus, $\mathcal{R}_k(x,u(0:k-1))$ in~\eqref{eq:PRS_design} is a PRS according to Definition~\ref{def:PRS}, i.e., Assumption~\ref{asm:PRS} holds. 
Assumption~\ref{asm:PRS_nom} follows due to the structure of $\mathcal{R}_k$ in~\eqref{eq:PRS_design}. 
\end{proof}
Theorem~\ref{thm:PRS} provides a simple recursive formula to compute a probabilistic reachable set for the nonlinear stochastic system~\eqref{eq:sys} using a nominal simulation $z(k)$ and a scaled ellipsoid $\mathbb{D}_k$. 
This result is valid for inputs computed by any causal policy and thus also for the inputs generated by the SMPC scheme. 
We highlight that the shape of the PRS $\mathcal{R}_k$ computed through stochastic contraction metrics is analogous to the robust reachable sets computed through robust contraction metrics in~\cite{zhao2022tube,singh2023robust,sasfi2023robust}. 
The main difference is that $\bar{w}$ is proportional to the variance of the stochastic noise $w(k)$, while existing results utilize uniform worst-base bounds on the magnitude of the noise $w(k)$. 
In case Assumption~\ref{asm:disturbance} is generalized to time-varying co-variance bounds $\Sigma_{\mathrm{w}}(k)$, Theorems~\ref{thm:CCM}--\ref{thm:PRS} would yield $\sigma_{\mathrm{x},k}=\sum_{j=0}^{k-1}\rho^{k-j}\frac{\bar{w}(j)}{1-p}$, i.e., a scaling accounting for the time-varying noise characteristics.

%!TEX root = main.tex
\section{Overall algorithm}
\label{sec:algorithm}
The theoretical analysis in Section~\ref{sec:SNMPC_tractable} relies on PRS (Asm.~\ref{asm:PRS}, \ref{asm:PRS_nom}), tightened constraints $\bar{\mathbb{X}}_k$, and a terminal cost $V_{\mathrm{f}}$ and terminal set $\mathbb{X}_{\mathrm{f}}$ (Asm.~\ref{asm:terminal}). 
Section~\ref{sec:CCM} showed how we can constructed PRS satisfying Assumptions~~\ref{asm:PRS}--\ref{asm:PRS_nom} using contraction metrics and Appendix~\ref{app_CCM} provides additional details regarding their optimization.
In the following, we discuss how the other two components are computed before summarizing the overall design.

\subsection{Constraint tightening}
\label{sec:algorithm_tightening}
The following proposition provides a simple formula to compute the tightened constraint set $\bar{\mathbb{X}}_k$. 
\begin{proposition}
\label{prop:tightening}
Suppose the conditions from Theorem~\ref{thm:PRS} and Assumption~\ref{asm:setup_regularity}\ref{asm:setup_regularity_constraint} hold. 
Then, for all $k\in\mathbb{I}_{\geq 0}$, the tightened constraints 
\begin{align}
&\bar{\mathbb{X}}_k=\left\{x|~h_j(x)+c_j\sqrt{\sigma_{\mathrm{x},k}}\leq 0,~j\in\mathbb{I}_{[1,r]}\right\}, \nonumber\\
\label{eq:tightened_constraints_prop}
&c_j=\max_{x\in\mathbb{R}^n}\left\|\underline{M}^{-1/2}\left.\dfrac{\partial h_j}{\partial x}\right|_x^\top\right\|,~j\in\mathbb{I}_{[1,r]}
\end{align}
satisfy $\bar{\mathbb{X}}_k\oplus\mathbb{D}_k\subseteq\mathbb{X}$. 
\end{proposition}
\begin{proof}
Note that $\bar{\mathbb{X}}_k\oplus\mathbb{D}_k\subseteq\mathbb{X}$ is equivalent to $x\in\mathbb{X}$ for any $z\in\bar{\mathbb{X}}_k$ and any $\|x-z\|_{\underline{M}}\leq \sqrt{\sigma_{\mathrm{x},k}}$ using~\eqref{eq:PRS_design}. 
Denote $\gamma(s)=z+s(x-z)$, $s\in[0,1]$ as the line connecting some $x$ an $z$. 
Similar to~\cite[Prop.~5]{sasfi2023robust}, the mean value theorem yields
\begin{align*}
h_j(x)\leq& h_j(z)+\int_0^1 \left. \dfrac{\partial h}{\partial x}\right|_{\gamma(s)}\left.\dfrac{\partial \gamma}{\partial s}\right|_s\mathrm{d}s\\
=& h_j(z)+\int_0^1\left. \dfrac{\partial h}{\partial x}\right|_{\gamma(s)}\underline{M}^{-1/2}\underline{M}^{1/2}(x-z) \mathrm{d}s\\
\leq &h_j(z)+\int_0^1\left. \left\|\underline{M}^{-1/2}\dfrac{\partial h}{\partial x} \right|^\top_{\gamma(s)}\right\| \left\|\underline{M}^{1/2}(x-z)\right\| \mathrm{d}s\\
\stackrel{\eqref{eq:tightened_constraints_prop}}{\leq}& h_j(z)+c_j\int_0^1\|x-z\|_{\underline{M}} d\mathrm{s}\\
= &h_j(z)+c_j\|x-z\|_{\underline{M}}\leq h_j(z)+c_j\sqrt{\sigma_{\mathrm{x},k}}\stackrel{\eqref{eq:tightened_constraints_prop}}{\leq} 0, 
\end{align*}
for all $j\in\mathbb{I}_{[1,r]}$ and thus $x\in\mathbb{X}$ (Asm.~\ref{asm:setup_regularity}\ref{asm:setup_regularity_constraint}). 
\end{proof}
In the special case of polytopic constraints, Proposition~\ref{prop:tightening} reduces to computing $\bar{\mathbb{X}}_k$ using the Pontryagin difference. 

\subsection{Design of terminal cost and set}
\label{sec_term}
The following proposition shows how to design a terminal set $\mathbb{X}_{\mathrm{f}}$ and a terminal cost $V_{\mathrm{f}}$ satisfying Assumption~\ref{asm:terminal}. 
\begin{proposition}
\label{prop:term}
Suppose that Assumptions~\ref{asm:setup_regularity}\ref{asm:setup_regularity_cost}, \ref{asm:setup_regularity}\ref{asm:setup_regularity_origin}, and \ref{asm:setup_regularity}\ref{asm:setup_regularity_CCM} hold, the conditions in Corollary~\ref{cor:CCM_stoch} are satisfied, 
 and $\mathbb{D}_k\subseteq \mathbb{X}$, $k\in\mathbb{I}_{\geq 0}$.
Then, there exist constants $\alpha_{\mathrm{f}}\geq 0,c_{\mathrm{f}}> 0$, such that Assumption~\ref{asm:terminal} holds with $V_{\mathrm{f}}(x)=\|x\|_P^2$, $P=c_{\mathrm{f}}M$, $\mathbb{X}_{\mathrm{f}}=\{x|~\|x\|_M^2\leq \alpha_{\mathrm{f}}\}$, and $u_{\mathrm{f}}=0$. 
\end{proposition}
\begin{proof}
Analogous to the proof of~\eqref{eq:increm_bounds_2} in Theorem~\ref{thm:CCM}, for any $x\in\mathbb{R}^n$: 
\begin{align}
\label{eq:CLF_CCM}
&\|f(x,0,0)\|_M^2=\|f(x,0,0)-f(0,0,0)\|_M^2\nonumber\\
\leq& \rho\|x-0\|_M^2=\rho\|x\|_M^2,
\end{align}
where we used $f(0,0,0)=0$ (Asm.~\ref{asm:setup_regularity}).
The positive invariance condition (Asm.~\ref{asm:terminal}~\ref{term_invar}) holds since $\mathbb{X}_{\mathrm{f}}$ is a sublevel set of the Lyapunov function $\|x\|_M^2$. \\
Constraint satisfaction (Asm.~\ref{asm:terminal}~\ref{term_con}) holds by choosing $\alpha_{\mathrm{f}}\geq 0$ as the largest constant satisfying $\mathbb{X}_{\mathrm{f}}\subseteq \bar{\mathbb{X}}_k$, $k\in\mathbb{I}_{\geq 0}$, 
with $\alpha_{\mathrm{f}}\geq 0$ since $0\in\bar{\mathbb{X}}_k=\mathbb{X}\ominus\mathbb{D}_k$, $k\in\mathbb{I}_{\geq 0}$, by assumption. \\
Note that $ \|x\|_Q^2\leq \lambda_{\max}(Q,M)\|x\|_M^2$ where $\lambda_{\max}(Q,M)>0$ is the maximal generalized eigenvalue of $Q$ w.r.t. $M$. 
The terminal cost condition (Asm.~\ref{asm:terminal}~\ref{term_CLF}) holds with 
\begin{align*}
&V_{\mathrm{f}}(f(x,0,0))=c_{\mathrm{f}}\|f(x,0,0)\|_M^2\stackrel{\eqref{eq:CLF_CCM}}{\leq} c_{\mathrm{f}}\rho\|x\|_M^2\\
\leq &c_{\mathrm{f}}\|x\|_M^2-\dfrac{(1-\rho)c_{\mathrm{f}}}{\lambda_{\max}(Q,M)}\|x\|_Q^2
= V_{\mathrm{f}}(x)-\ell(x,0),
\end{align*}
where the last equality uses $c_{\mathrm{f}}:=\frac{\lambda_{\max}(Q,M)}{1-\rho}> 0$.
\end{proof}
\subsection{Summary of theoretical properties}
The overall offline design and online implementation are summarized in Algorithm~\ref{alg}.
\begin{algorithm}
\caption{SMPC: Offline design and online implementation}
\label{alg}
\begin{algorithmic}[t]
\State Given model $f$, constraints $\mathbb{X},\mathbb{U}$, variance bound $\Sigma_{\mathrm{w}}$, cost matrices $Q,R\succ 0$, probability level $p\in(0,1)$. 
\State \textbf{Offline design}
\State Compute contraction metric $M,\rho,\bar{w}$ (App.~\ref{app_CCM}). 
\State Compute tightened constraints $\overline{\mathbb{X}}_k$ (Prop.~\ref{prop:tightening}). 
\State Design terminal cost $V_{\mathrm{f}}$ and terminal set $\mathbb{X}_{\mathrm{f}}$ (Prop.~\ref{prop:term}).
\State \textbf{Online operation}
\State Initialize $z(0)=x(0)=x_0$. 
\For{$k\in\mathbb{I}_{\geq 0}$}
	\State Measure the state $x(k)$ from system~\eqref{eq:sys}.
	\State Solve Problem~\eqref{eq:SNMPC} using $x(k)$, $z(k)$.
	\State Apply the control input $u(k)=\mathbf{u}^\star_{k|k}$.
	\State Simulate nominal state $z(k+1)$ using~\eqref{eq:nominal}. 
\EndFor
\end{algorithmic}
\end{algorithm} \\
The offline computation involves the construction of the contraction metric, tightening of the constraints, and the computation of the terminal cost/set. 
The following corollary demonstrates that this procedure satisfies the posed conditions on the PRS (Asm.~\ref{asm:PRS}--\ref{asm:PRS_nom}) and the terminal ingredients (Asm.~\ref{asm:terminal}). 
\begin{corollary}
\label{cor:design}
Suppose Assumptions~\ref{asm:disturbance} and \ref{asm:setup_regularity} hold. 
Consider any feasible solution $(W,X,\rho)$ to Problem~\eqref{eq:CCM_LMI} in Appendix~\ref{app_CCM} with $\tr{X}\leq (1-\rho)(1-p)$. 
Then, $\underline{M}=W^{-1}$, $\bar{w}=\tr{X}$, $u_{\mathrm{f}}=0$, $V_{\mathrm{f}}=c_{\mathrm{f}}\|x\|_M^2$, $\mathbb{X}_{\mathrm{f}}=\{x|~\|x\|_M^2\leq \alpha_{\mathrm{f}}\}$ with $c_{\mathrm{f}}>0$, $\alpha_{\mathrm{f}}\geq 0$ from Proposition~\ref{prop:term} and $\mathcal{R}_k$, $\mathbb{D}_k$ from Theorem~\ref{thm:PRS} satisfy Assumptions~\ref{asm:PRS}, \ref{asm:PRS_nom}, and \ref{asm:terminal}. 
\end{corollary}
\begin{proof}
Proposition~\ref{prop:CCM_LMI} with $(W,X,\rho)$ from Problem~\eqref{eq:CCM_LMI} ensures that the conditions in Corollary~\ref{cor:CCM_stoch} hold and thus Assumptions~\ref{asm:PRS} and \ref{asm:PRS_nom} hold. 
Furthermore, the posed bound on $\tr{X}$ and $h_j(0)=-1$ ensures that $0\in\bar{\mathbb{X}}_k$, $k\in\mathbb{I}_{\geq 0}$. 
Thus, the conditions in Proposition~\ref{prop:term} hold and the resulting terminal cost and set satisfy Assumption~\ref{asm:terminal}. 
\end{proof}
This design ensures that the conditions in Theorem~\ref{thm:SMPC} are satisfied. 
Hence, Algorithm~\ref{alg} provides the desired closed-loop guarantees:~\ref{enum_desired_feas},~\ref{enum_desired_chance},~\ref{enum_desired_cost}. 
The online operation  to compute the control input only requires solving a finite-dimensional deterministic optimization problem~\eqref{eq:SNMPC}.

\section{Discussion}
\label{sec:discussion}
In this section, we contrast the proposed SMPC formulation (Sec.~\ref{sec:SNMPC_shrinking}--\ref{sec:SNMPC_tractable}) and the derived PRS (Sec.~\ref{sec:CCM}) to existing methods in terms of conservatism and computational complexity. 
\subsection{SMPC formulations}
\label{sec:discussion_related}
The proposed SMPC scheme (Sec.~\ref{sec:SNMPC_tractable}) addresses nonlinear stochastic systems with unbounded noise and guarantees: \ref{enum_desired_feas}: recursive feasibility, \ref{enum_desired_chance}: closed-loop chance constraint satisfaction, and \ref{enum_desired_cost}: closed-loop performance. 
The SMPC is characterized by a finite-dimensional deterministic optimization problem~\eqref{eq:SNMPC}, which can be equivalently written as a nominal MPC with a state $(x,z)\in\mathbb{R}^{2n}$ and control input $u\in\mathbb{R}^m$. 
Hence, the computational complexity is only marginally increased compared to a nominal MPC; see also the numerical comparison in Section~\ref{sec:num_compare}.
In the following, we contrast the proposed design and the closed-loop properties to existing approaches. 
\subsubsection*{Shrinking-horizon case}
Section~\ref{sec:SNMPC_shrinking} first addresses the finite-horizon problem with a shrinking horizon approach.
For this problem, an alternative solution has been proposed in~\cite{wang2021recursive}, which also ensures \ref{enum_desired_feas} and \ref{enum_desired_chance}. 
However, theoretical guarantees in~\cite{wang2021recursive} require exact evaluation of different predictive distributions for re-conditioning and a tractable implementation is only provided for the special case of linear systems using scenario theory~\cite{campi2019scenario}. 
In contrast, the formulation in Section~\ref{sec:SNMPC_shrinking} can leverage general PRS (Def.~\ref{def:PRS}) and does not require exact evaluation of probabilities online. 
\subsubsection*{Relation to existing linear SMPC}
The online operation (Alg.~\ref{alg}) and SMPC formulation (Problem~\eqref{eq:SNMPC}) are inspired by the \textit{indirect-feedback} SMPC paradigm~\cite{Hewing2020indirect,hewing2019scenario,muntwiler2022lqg,mark2021data,arcari2023stochastic}: 
The constraints are imposed on the nominal state $z(k)$, while the measured state $x(k)$ is only used in the cost function. 
This decoupling ensures recursive feasibility, even for arbitrary large noise realizations. 
In~\cite{Hewing2020indirect,hewing2019scenario,muntwiler2022lqg,mark2021data,arcari2023stochastic}, linearity of the dynamics is leveraged to define error dynamics which evolve completely independent of the SMPC, thus facilitating the construction of PRS for the error. 
On the other hand, the presented analysis leverages parametrized PRS that are valid for any causal policy generated by the SMPC. 
This relaxation of the independence assumption expands the applicability of this paradigm to general nonlinear stochastic optimal control problems, while still containing the approach~\cite{Hewing2020indirect} as a special case. 

\subsubsection*{Recursively feasible nonlinear SMPC formulations}
The works in~\cite{lorenzen2019stochastic,schluter2020constraint,mcallister2022nonlinear} provide nonlinear SMPC schemes that ensure recursive feasibility and closed-loop chance constraint satisfaction. 
These approaches assume a uniform bound on the worst-case noise realizations to robustly ensure recursive feasibility using specific PRS constructions. 
Hence, they cannot leverage general PRS (Def.~\ref{def:PRS}) or address general noise distributions (Asm.~\ref{asm:disturbance}). 
As a result, these existing approaches are conservative in case of rare large noise, see also the qualitative comparisons in~\cite{koehler2023stochastic} and the numerical comparison in Section~\ref{sec:num_compare}. 
The proposed approach avoids these issues by using general PRS and a nominal state $z(k)$ to formulate the optimization problem.

\subsubsection*{Expected cost bounds}
Theorem~\ref{thm:SMPC} shows that the asymptotic expected cost is bounded by a function of the variance. 
Similar bounds on the expected cost are well established in the SMPC literature~\cite{chatterjee2014stability,Hewing2020indirect,lorenzen2019stochastic,mcallister2022nonlinear}. 
However, there are key differences in the underlying analysis, designs, and assumptions. 
The proposed analysis utilizes the specific structure of the quadratic cost and the (global) properties of the terminal cost to derive this expected cost bound. 
Compared to~\cite{chatterjee2014stability,lorenzen2019stochastic,mcallister2022nonlinear}, our analysis is complicated by the fact that the optimal cost $\mathcal{J}_N^\star$ is not only a function of the state $x(k)$ and that $x(k)$ may be unbounded. 
In particular, the analysis in~\cite{mcallister2022nonlinear} relies on uniform bounds on the state $x(k)$. 
While~\cite{chatterjee2014stability} also uses a global terminal cost to derive a global bound on the optimal cost, this step requires extra care in the proposed analysis due to the difference in the SMPC formulation, see derivation of~\eqref{eq:cost_upper_bound}. 
While~\cite{chatterjee2014stability,lorenzen2019stochastic,mcallister2022nonlinear,Hewing2020indirect} minimize the expected cost, the proposed formulation and analysis (Sec.~\ref{sec:SNMPC_tractable}) considers a certainty-equivalent cost, which avoids computationally intensive sampling-based approximations. 
As shown in~\cite{messerer2024fourth}, the benefits of minimizing the expected cost compared to a certainty-equivalent cost are often negligible, especially for smooth quadratic costs. 

\subsection{Computation of probabilistic reachable sets}
\label{sec:discussion_PRS}
In the following, we contrast the proposed PRS computation to sampling-based approximation and also discuss extensions to closed-loop predictions. 
\subsubsection*{Offline constructed PRS}
Section~\ref{sec:CCM} provides a design of PRS using contraction metrics. This construction computes  the contraction metric $M$ offline using LMIs while the online evaluation of the PRS only requires simulating a nominal state $z(k)\in\mathbb{R}^n$. 
This is comparable to existing \textit{robust} reachability techniques using contraction metrics~\cite{zhao2022tube,singh2023robust,sasfi2023robust}. 
For large to medium scale problems, a direct application of these LMI designs may face scalability issues and instead structured parametrizations can be used to maintain scalability~\cite{shiromoto2018distributed}. 
\subsubsection*{Sampling-based approximations}
A common alternative to compute PRS are sampling-based approximations, i.e., the nonlinear system is simulated for $N_{\mathrm{s}}\in\mathbb{I}_{\geq 1}$ different random realizations of the noise $w$ and the union of these trajectories approximates the PRS.
Key benefits of such techniques are ease of implementation, absence of regularity conditions (Asm.~\ref{asm:setup_regularity}), and asymptotic exactness\footnote{%
Exactness also requires exclusion of the worst-case $(1-p)\cdot N_{\mathrm{s}}$ samples, which can be formulated as a (large-scale) mixed-integer problem~\cite{blackmore2010probabilistic}.} for $N_{\mathrm{s}}\rightarrow\infty$. 
However, sampling-based approaches also have a number of disadvantages: 
\begin{itemize}
\item (i) computational complexity increases with the number of samples $N_{\mathrm{s}}$;
\item (ii) theoretical guarantees require choosing $N_{\mathrm{s}}$ large enough~\cite{garatti2024non};
\item (iii) computing PRS $\mathcal{R}_k$ for $k\rightarrow\infty$ is intractable, but is required for recursive feasibility;
\item (iv) implementation requires the noise distribution to be time-invariant and known.
\end{itemize}
The numerical comparison in Section~\ref{sec:num} also highlights some of these drawbacks.
In the special case of \textit{linear} stochastic systems, such sampling-based techniques are well established using scenario theory~\cite{wang2021recursive,blackmore2010probabilistic,hewing2019scenario,campi2019scenario,Lorenzen2017tightening,aolaritei2023wasserstein} and the  difficulties can sometimes be mitigated.
\subsubsection*{Closed-loop predictions}
In the following, we discuss the generalization to closed-loop predictions with relaxed probabilistic input constraints $\prob{u(k)\in\mathbb{U}}\geq p$, $k\in\mathbb{I}_{\geq 0}$ (cf. Remark~\ref{rk:input}). 
Suppose we have designed a feedback $u=\kappa(x,z,v)$, such that for any $x,z\in\mathbb{R}^n$, $v\in\mathbb{U}$: 
\begin{align}
&\mathbb{E}_{w(k)}[V_\delta(f(x,\kappa(x,z,v),w),f(z,v,0))]\leq \rho V_\delta(x,z)+\bar{w},\nonumber\\
\label{eq:CCM_kappa_2}
&\|\kappa(x,z,v)-v\|^2\leq L_\kappa\|x-z\|_{\underline{M}}^2,~L_\kappa\geq 0.
\end{align}
Such a feedback can be jointly synthesized with the contraction metric $M$ by adapting Theorem~\ref{thm:CCM} and Proposition~\ref{prop:CCM_LMI} from Appendix~\ref{app_CCM}, see~\cite{zhao2022tube,singh2023robust,sasfi2023robust}.
Similar to Theorem~\ref{thm:PRS}, for any causal policy $v(k)=\pi_k(w(0:k-1))$,
the closed-loop system
\begin{align*}
z(k+1)=&f(z(k),v(k),0),~z(0)=x(0)=x_0,\\
x(k+1)=&f(x(k),u(k),w(k)),~u(k)=\kappa(x(k),z(k),v(k))
\end{align*}
satisfies
\begin{align*}
\prob{|x(k)-z(k)\|_{\underline{M}}^2\leq \sigma_{\mathrm{x},k},~\|u(k)-v(k)\|^2\leq L_\kappa\sigma_{\mathrm{x},k}}\geq p
\end{align*}
with $\sigma_{\mathrm{x},k}$ from~\eqref{eq:PRS_design}. 
By including the feedback $\kappa$ in the prediction, we can deal with unstable systems and reduce the conservatism of the PRS.

%!TEX root =main.tex
\section{Numerical example}
\label{sec:num}
The following numerical example demonstrates the applicability of the proposed framework to nonlinear stochastic systems and highlights the computational tractability and theoretical guarantees compared to existing SMPC methods. 
The code is available online.\footnote{https://gitlab.ethz.ch/ics/SMPC-CCM}
All statistical quantities (expected cost, probability of constraint violation) are approximated by using $10^5$ independently sampled noise sequences.
\begin{figure}[t]
\centering
\includegraphics[width=0.45\textwidth]{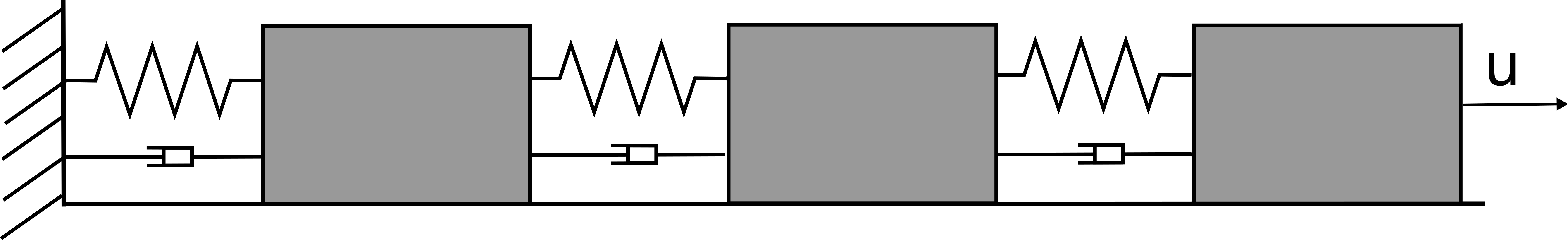}
\includegraphics[width=0.3\textwidth]{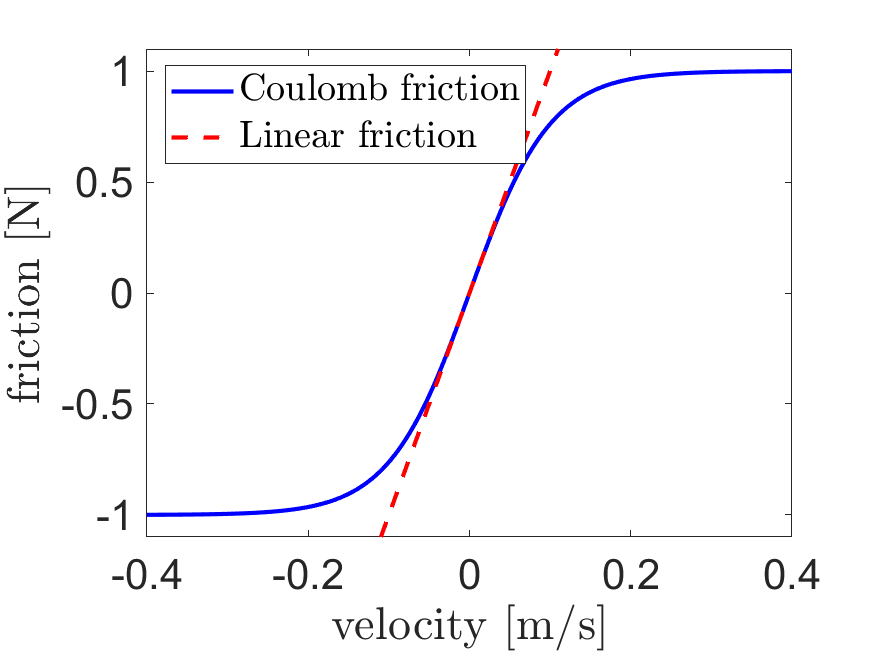}
\caption{Top: Considered system consisting of a chain of three mass-spring-dampers with random wind forces $w\in\mathbb{R}^3$ and actuation $u\in\mathbb{R}$. 
Bottom: Characteristic curve of Coulomb friction compared to linear friction.}
\label{fig:sketch_mass_spring}
\end{figure}
\subsection{Setup}
We consider a chain of three mass-spring-dampers with nonlinear Coulomb friction, shown in Figure~\ref{fig:sketch_mass_spring}. 
The model equations are given by
\begin{align*}
m\dot{v}_1=&k(p_2-2p_1)+d(v_2-2v_1)-f_{\mathrm{nl}}(v_1)+w_1\\
m\dot{v}_2=&k(p_1+p_3-2p_2)+d(v_1+v_3-2v_2)-f_{\mathrm{nl}}(v_2)+w_2\\
m\dot{v}_3=&k(p_2-p_3)+d(v_2-v_3)-f_{\mathrm{nl}}(v_3)+w_3+u
\end{align*}
and $\dot{p}=v$ with positions $p$, velocities $v$, state $x=[p;v]\in\mathbb{R}^6$, control input $u\in \mathbb{R}^1$, process noise $w\in\mathbb{R}^3$, mass $m=5$, spring constant $k=2$, damping constant $d=1$, and the Coulomb friction $f_{\mathrm{nl}}(v)=\tanh(10\cdot v)$. 
The linearized system is under-damped with the slowest mode having a time constant of $50$~[s] and fast oscillations at a frequency of about 1~[Hz]. 
The noise $w$ has a covariance of $\Sigma_{\mathrm{w}}=10^{-3}\cdot I_3$ with a later specified distribution. 
The model is discretized using Euler with a sampling time of $dt=0.25$~[s]. 
The actuator force is bound by $\mathbb{U}=[-10^2,10^2]$~[N]. 
We consider chance constraints~\eqref{eq:prob_state} with $p=95\%$ and the polytope $\mathbb{X}=\{[p;v] \mid\min\{p_1,p_2-p_1,p_3-p_2\}\geq 1, \|v\|_\infty\leq 2\}$.

\subsubsection*{Offline design}
The offline design is performed according to Corollary~\ref{cor:design} with a contraction rate of $\rho\approx\exp(-dt/50).$ 
Executing the overall offline design (Alg.~\ref{alg}) required about two seconds on a standard laptop, where the LMIs were formulated with YALMIP~\cite{lofberg2004yalmip} and solved using SeDuMi1.3 and Matlab.

\subsection{Probabilistic reachable set}
\begin{figure}[t]
\centering 
 \includegraphics[width=0.4\textwidth]{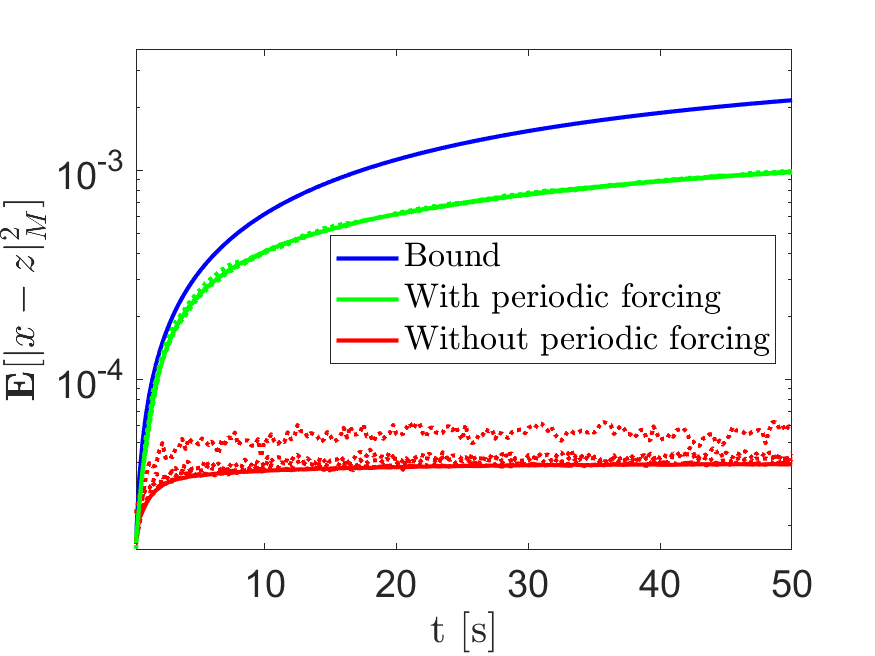}
 \includegraphics[width=0.4\textwidth]{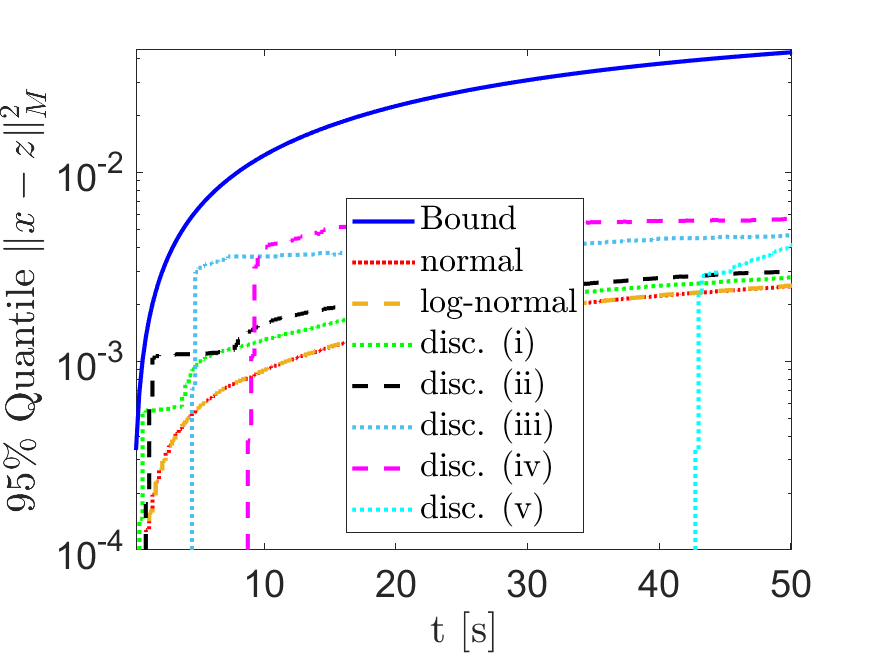}
 \caption{Open-loop simulations: 
Top: Expected (weighted) squared difference between nominal prediction $z$ and perturbed real state $x$: derived bound (blue, solid); empirical mean with periodic forcing (green) and without any forcing (red) for normal distributions (solid, thick) and different discrete distributions (dotted). 
Bottom: Derived bound (blue, solid) on the $p=95\%$ quantile of $\|x-z\|_M^2$, compared to empirical quantiles for the different distributions (dotted and dashed) with periodic forcing}.
 \label{fig:expected_PRS}
\end{figure} 
First, we consider the proposed PRS design using stochastic contraction metrics (Sec.~\ref{sec:CCM}) and study the conservatism depending on the noise distribution and operating conditions. 
We consider a zero initial condition. 
Two different input signals $u$ are applied to the system: 
\begin{enumerate}[label=\alph*)]
\item periodic forcing with amplitude $1$~[kN] and period $12.5$~[s], operating mostly in the saturated region of the Coulomb friction (cf. Fig.~\ref{fig:sketch_mass_spring});
\item zero input, operating close to the origin with almost linear dynamics.
\end{enumerate}
We examine seven different noise distributions $\mathcal{Q}_{\mathrm{w}}$ with zero mean and covariance $\Sigma_{\mathrm{w}}$:
\begin{itemize}
\item Normal distribution $w\sim\mathcal{N}(0,\Sigma_{\mathrm{w}})$;
\item Log-normal distribution shifted to zero mean;
\item discrete-distributions\footnote{%
For each coordinate $i$, $w_i\in\{-c,0,c\}$ with $c$ such that $\mathbb{E}_w[ww^\top]=\Sigma_{\mathrm{w}}$. 
These distributions are the (point-wise) worst-case when applying the Markov inequality to bound $|w|^2$ using only $\mathbb{E}[w^2]$.} with \\
$\prob{w(k)=0}\in\{99\%,99.5\%,99.9\%,99.95\%,99.99\%\}$.
\end{itemize}
Figure~\ref{fig:expected_PRS} shows the bound on expected squared error derived in Theorem~\ref{thm:CCM} and the empirical results for different input signals. We see a large disparity in the error (factor $50$) for the two different input signals, which highlights the significant effect of the nonlinearities.
This also implies that na\"ive tuning of constraint tightening based on offline sampling (cf., e.g., \cite[Sec.~5.5]{mayne2016robust}) may not generalize well for nonlinear dynamics. 
In contrast, the derived bound is valid independent of the specific input $u$ applied to the system or the exact noise distribution $\mathcal{Q}_{\mathrm{w}}$. 
Figure~\ref{fig:expected_PRS} studies the conservatism of the probabilistic reachable set derived in Theorem~\ref{thm:PRS}, where we only focus on the more critical case of periodic forcing. 
Here, the exact distribution $w(k)\sim\mathcal{Q}_{\mathrm{w}}$ has also a significant effect. 
This implies that any method that directly utilizes the distribution $\mathcal{Q}_{\mathrm{w}}$, such as sampling-based approaches, can only ensure safe operation if the underlying distribution $\mathcal{Q}_{\mathrm{w}}$ is exactly known. 
In contrast, the derived theoretical results~(Thm.~\ref{thm:PRS}) provide an upper bound on the empirical quantile for all distributions, highlighting the distributional robustness but also the potential conservatism of the derived result.

\begin{table*}
\begin{tabular}{r|l|l|l|l|l|l|l|l|l|l|l|l|l}
Method&Nominal&Proposed&Bounded&\multicolumn{7}{c}{Sampling $N_{\mathrm{s}}\in\{10,20,50,100,200,500,1000\}$}\\ \hline
Computation time [ms]&27&34&30&139&307&990&$2.2\cdot 10^3$&$4.9\cdot 10^3$&$1.5\cdot 10^4$&$3.1\cdot 10^4$\\
$\max_k\prob{x_k\notin\mathbb{X}}[\%]$& 14.1&0.0&0.0&14.1&11.5&6.0&2.0&0.8&0.6&0.2\\
$\mathbb{E}_{w(\cdot)}\left[\mathcal{J}_N(x(\cdot),u(\cdot))\right]$&1639&1670&1822&1639&1639&1644&1647&1648&1652&1650
\end{tabular}
\vspace{1mm}
\caption{Comparison of nominal MPC, the proposed PRS computation, PRS computations leveraging bounded support, and approximations with $N_{\mathrm{s}}$ samples for a finite-horizon stochastic optimal control problem, showing computational time, maximal probability of constraint violation, and average cost.}
\label{tab:comp}
\end{table*}

\subsection{Comparison}
\label{sec:num_compare}
{We consider an open-loop finite-horizon control problem. 
To ensure applicability of existing approaches, we consider the discrete distribution (ii) from Figure~\ref{fig:expected_PRS}, which has bounded support, and we assume that this distribution is precisely known. 
We compare the following approaches:
\begin{itemize}
\item \textit{Nominal:} A nominal MPC formulation that neglects the presence of noise;
\item \textit{Proposed:} The proposed approach, with PRS computed using contraction metrics (Thm.~\ref{thm:PRS});
\item \textit{Bounded:} SMPC formulations~\cite{schluter2020constraint,bonzanini2019tube,mcallister2022nonlinear} using bounded support of the noise to compute PRS\footnote{%
These methods robustly enforce $\prob{x_{k+1|k}\in\mathbb{X}|x_k}\geq p$ by combing a robust reachable set with a disturbance set containing $p\%$ of the noise. We implemented the methods in~\cite{schluter2020constraint,bonzanini2019tube} by computing a robust reachable set using the contraction metric $M$.}; 
\item \textit{Sampling:} Sampling-based approximation~\cite{campi2019scenario}, using $N_{\mathrm{s}}\in\{10,20,50,100,200,500,1000\}$ samples\footnote{%
We sample noise sequences and enforce that all trajectories satisfy the constraints and minimize the average cost. 
Probabilistic constraint satisfaction holds with high confidence if a sufficiently large number of samples $N_{\mathrm{s}}$ is used, cf. scenario theory~\cite{garatti2024non}.}. 
\end{itemize}
All schemes are solved with IPOPT \cite{wachter2006implementation} formulated in CasADi \cite{andersson2019casadi} in Matlab on a standard laptop using a zero warm-start. 
We consider a prediction horizon of $N=15$ and the initial condition $x(0)=[0_2;10;0_3]$. 
The results can be seen in Table~\ref{tab:comp}. 
The computational demand of the sampling-based approaches increases with the number of samples, resulting in computation times that are orders of magnitude larger than the proposed approach. 
The proposed approach and the SMPC formulation leveraging bounded noise support both satisfy the constraints with the desired probability~\eqref{eq:prob_state}, while the nominal formulation fails to meet the probabilistic constraint. 
Sampling-based approximations empirically satisfy the probabilistic constraints~\eqref{eq:prob_state} if the number of samples is sufficiently large. 
The large expected cost of the SMPC schemes~\cite{schluter2020constraint,bonzanini2019tube,mcallister2022nonlinear} demonstrates the conservatism of their PRS construction, which leverages the worst-case noise bounds.}

\subsection{Closed-loop simulations}
Lastly, we conduct closed-loop simulations using the proposed SMPC formulation with Problem~\eqref{eq:SNMPC}. 
We consider the discrete distribution (iv) from Figure~\ref{fig:expected_PRS}, which yields the largest prediction error. 
The sampling-based approximations are not applicable since we consider operation over infinite time (cf. discussion in Sec.~\ref{sec:discussion}). 
The SMPC formulation leveraging bounded noise support~\cite{schluter2020constraint,bonzanini2019tube,mcallister2022nonlinear} are infeasible for this problem due to the large magnitude of the worst-case noise. 
Figure~\ref{fig:closedloop} shows the closed-loop results of the proposed method. 
The expected cost decreases during closed-loop operation and converges to a small constant, in accordance with Theorem~\ref{thm:SMPC}. 
The variations in the applied input highlights that feedback is generated through re-optimization. 
Lastly, although we operate close to the constraints, constraint violations are rare and the posed chance constraints~\eqref{eq:prob_state} are met with: 
$\max_{\in\mathbb{I}_{\geq 0}}k~ \prob{x(k)\notin\mathbb{X}}=0.17\%< 1-p=5\%.$

\begin{figure}[th]
\centering
 \includegraphics[width=0.4\textwidth]{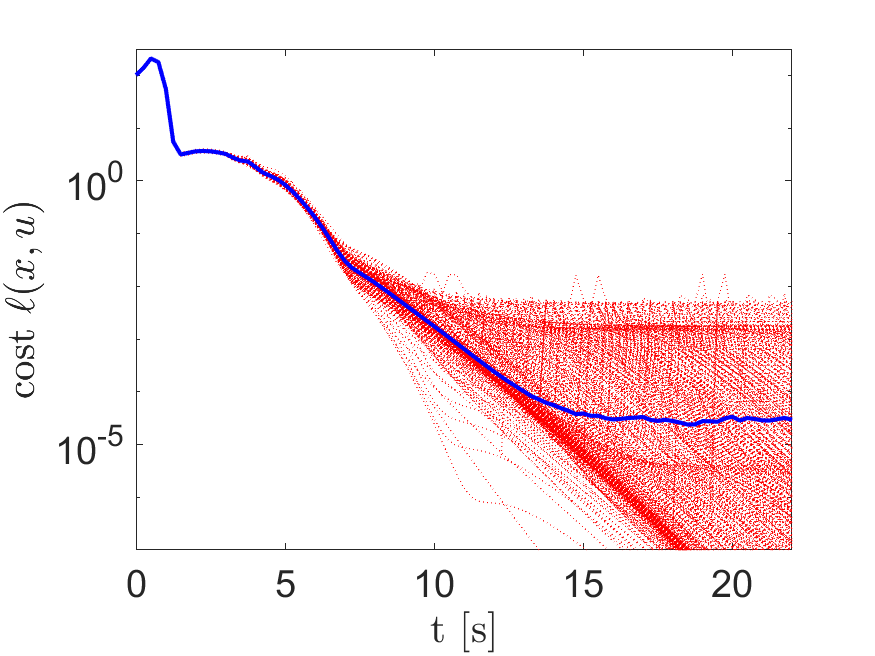}
 \includegraphics[width=0.4\textwidth]{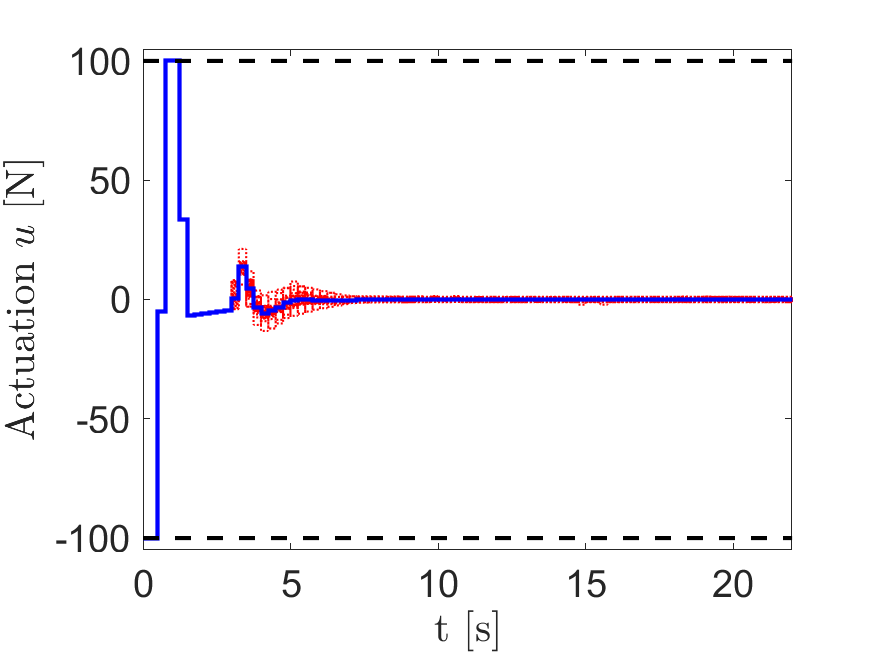}
 \includegraphics[width=0.4\textwidth]{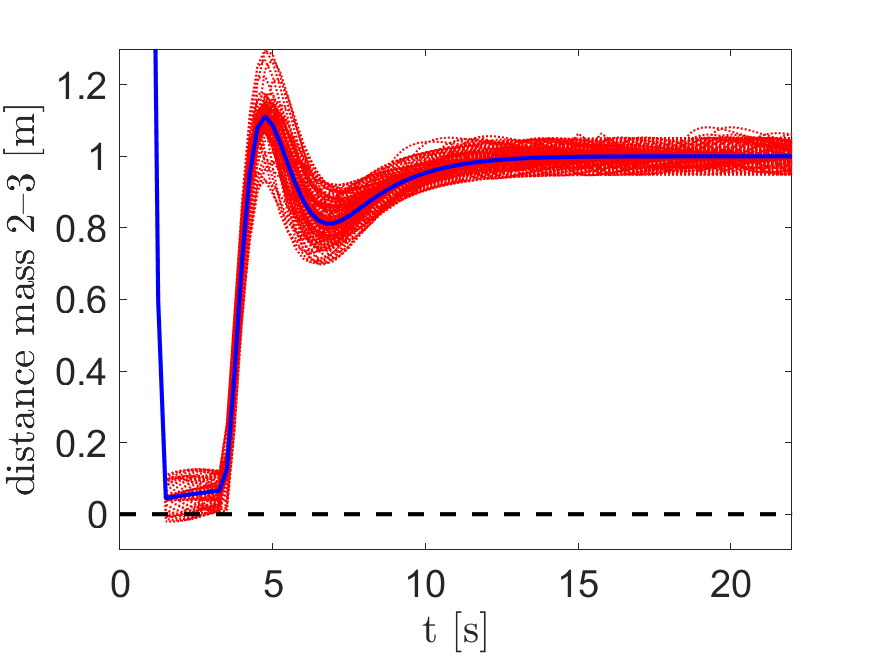}
\caption{Closed-loop system using proposed SMPC: random realizations (red, dotted) and mean (blue, solid).
Top: Closed-loop cost $\ell$. 
Middle: Applied actuation $u$ with input constraint (black, dashed). 
Bottom: Constraint $p_3-p_2+1\geq 0$ corresponding to a minimal distance between second and third mass (black, dashed). 
}
\label{fig:closedloop}
\end{figure}

\subsubsection*{Summary} 
Overall, we have applied the proposed stochastic predictive control framework to a nonlinear stochastic system. 
The statistics of the prediction error depend strongly on the applied inputs and exact distribution of the noise. 
The derived theory provides bounds on the expected error and probabilistic containment that hold uniformly for all possible input sequences and all distributions adhering to the assumed covariance bound (Asm.~\ref{asm:disturbance}). 
Thus, the closed-loop system satisfies the posed chance constraints~\eqref{eq:prob_state}. 
Furthermore, the proposed formulation is recursively feasible and the expected cost decreases to a small constant. 
Compared to common sampling-based approximations, the proposed SMPC scheme provides theoretical guarantees with an online computational complexity that is only marginally increased compared to a nominal MPC. 
This is crucial for implementation on embedded hardware for many applications. 
%!TEX root = main.tex 
\section{Conclusion}
\label{sec:conclusion}
We presented a computationally tractable predictive control framework for nonlinear systems subject to unbounded stochastic noise. 
The resulting closed-loop system ensures recursive feasibility, chance constraint satisfaction, and an expected cost bound. 
The design leverages stochastic contraction metrics to design probabilistic reachable sets. 
We highlighted advantages of the proposed method in a numerical example involving nonlinear dynamics and unbounded stochastic noise.  
Future work aims at designing less conservative probabilistic reachable sets by leveraging more distributional information and accounting for parametric uncertainty. 
\appendix
%!TEX root = main.tex
\subsection{Design of stochastic contraction metric}
\label{app_CCM}
In the following, we discuss the computation of a contraction metric $M$ satisfying the conditions in Theorem~\ref{thm:CCM}. 
We first pose an optimization problem that minimize the resulting constraint tightening and then discuss how to formulate it as a finite-dimensional semi-definite program (SDP). 
For ease of exposition, we focus on the constant parametrization $M\in\mathbb{R}^{n\times n}$ from Corollary~\ref{cor:CCM_stoch} and comment on addressing state-dependent metrics $M(x)$ at the end. 

\subsubsection*{Minimize constraint tightening}
We formulate LMIs that minimize the constraint tightening in Proposition~\ref{prop:tightening} resulting from the PRS in Theorem~\ref{thm:PRS}. 
We denote the gradient related to the state constraints (cf. Asm.~\ref{asm:setup_regularity}\ref{asm:setup_regularity_constraint}) by $h_{\mathrm{x},j}(x)=\left.\dfrac{\partial h_j}{\partial x}\right|_x^\top\in\mathbb{R}^n$ and consider the following optimization problem:
\begin{subequations}
\label{eq:CCM_LMI}
\begin{align}
\label{eq:CCM_LMI_cost}
\inf_{W,X,\rho} &~\dfrac{1}{1-\rho}\tr{X} \\
\mathrm{s.t. }&
\label{eq:CCM_LMI_C}
\begin{pmatrix}
 W& Wh_{\mathrm{x},j}(x) \\
(W h_{\mathrm{x},j}(x))^\top& 1
\end{pmatrix}\succeq 0,~j\in\mathbb{I}_{[1,r]},\\ 
&\label{eq:CCM_LMI_rho}
\begin{pmatrix}
\rho W& (A(x,u)W)^\top\\
A(x,u)W& W
\end{pmatrix}\succeq 0,\\
&\label{eq:CCM_LMI_W_matrix}
\begin{pmatrix}
X& ( G\Sigma_{\mathrm{w}}^{1/2})^\top\\
 G\Sigma_{\mathrm{w}}^{1/2}& W
\end{pmatrix}\succeq 0,\\
&W\succ 0,~\rho\in[0,1),\\
& \forall (x,u)\in\mathbb{R}^n\times\mathbb{U}.\nonumber
\end{align}
\end{subequations}
\begin{proposition}
\label{prop:CCM_LMI}
Let Assumptions~\ref{asm:disturbance}, \ref{asm:setup_regularity}\ref{asm:setup_regularity_constraint}, and \ref{asm:setup_regularity}\ref{asm:setup_regularity_CCM} hold. 
Consider any feasible solution $(W,X,\rho)$ to Problem~\eqref{eq:CCM_LMI}. 
Then, the conditions in Corollary~\ref{cor:CCM_stoch} hold with $M=W^{-1}$, $\rho$, and $\bar{w}=\tr{X}$. 
Furthermore, for all $k\in\mathbb{I}_{\geq 0}$, the tightened constraints according to~\eqref{eq:tightened_constraints_prop} satisfy
\begin{align}
\label{eq:CCM_LMI_tightened_constraint}
\overline{\mathbb{X}}_k\supseteq \left\{x|~h_j(x)\leq -\sqrt{\frac{\tr{X}}{(1-\rho)(1-p)}},~j\in\mathbb{I}_{[1,r]}\right\}.
\end{align}
\end{proposition}
\begin{proof}
\textbf{Part I:} 
Applying a Schur complement to Condition~\eqref{eq:CCM_LMI_W_matrix} yields
\begin{align*}
X- \Sigma_{\mathrm{w}}^{1/2} G^\top M G\Sigma_{\mathrm{w}}^{1/2}\succeq 0,
\end{align*}
where we used $M=W^{-1}\succ 0$. 
Given that $\tr{C_1}\geq \tr{C_2}$ for any matrices $C_1\succeq C_2\succeq 0$, this implies
\begin{align*}
\bar{w}=&\tr{X}\geq \tr{\Sigma_{\mathrm{w}}^{1/2}G^\top M G\Sigma_{\mathrm{w}}^{1/2}}\\
=&\tr{\Sigma_{\mathrm{w}}G^\top M G},
\end{align*}
i.e., \eqref{eq:CCM_w_simplified} holds. 
The Schur complement of~\eqref{eq:CCM_LMI_rho} yields
\begin{align*}
\rho W - W A(x,u)^\top M A(x,u)W\succeq 0.
\end{align*}
Multiplying with $M=W^{-1}\succ 0$ from left and right yields~\eqref{eq:CCM_rho_simplified}. 
Thus, the conditions in Corollary~\ref{cor:CCM_stoch} hold. \\
\textbf{Part II:} 
Applying Schur complement to \eqref{eq:CCM_LMI_C} yields
\begin{align*}
0\leq 1-(W h_{\mathrm{x},j}(x))^\top W^{-1} W h_{\mathrm{x},j}(x)=1-\|W^{1/2}h_{\mathrm{x},j}(x)\|^2,
\end{align*}
which implies $c_j= \|M^{-1/2}h_{\mathrm{x},j}(x)\|\leq 1$ with $W^{1/2}=M^{-1/2}$. 
Hence, the constraint tightening formula~\eqref{eq:tightened_constraints_prop} from Proposition~\ref{prop:tightening} with $M=\underline{M}$ satisfies
\begin{align*}
c_j\sqrt{\sigma_{\mathrm{x},k}}\leq \sqrt{\dfrac{\bar{w}}{(1-\rho)(1-p)}}
= \sqrt{\dfrac{\tr{X}}{(1-\rho)(1-p)}},
\end{align*}
yielding the inner-approximation~\eqref{eq:CCM_LMI_tightened_constraint}.
\end{proof}
Inequality~\eqref{eq:CCM_LMI_tightened_constraint} shows that the objective in Problem~\eqref{eq:CCM_LMI} bounds the constraint tightening.
Thus, Problem~\eqref{eq:CCM_LMI} yields contraction metrics that minimize the constraint tightening in the SMPC formulation~\eqref{eq:SNMPC}, see also~\cite{zhao2022tube,tsukamoto2020robust} for similar procedures in robust MPC.
\subsubsection*{Finite-dimensional convex problem}
Next, we discuss how to solve Problem~\eqref{eq:CCM_LMI}. 
For a fixed constant $\rho$, Problem~\eqref{eq:CCM_LMI} is an SDP with a linear cost and LMI constraints.
Hence, $\rho$ is computed using a line-search or bi-section, as similarly suggested in~\cite{tsukamoto2020robust,zhao2022tube}. 

As is standard in contraction metrics~\cite{zhao2022tube,singh2023robust,sasfi2023robust}, Inequalities~\eqref{eq:CCM_LMI} need to be verified for all $(x,u)\in\mathbb{R}^n\times\mathbb{U}$, which is not directly computationally tractable.
Standard solutions to this problem include heuristic gridding or sums-of-squares programming~\cite{zhao2022tube,singh2023robust}. 
However, both approaches are difficult to apply for the considered global discrete-time conditions. 
Hence, we instead use a convex embedding~\cite{koehler2020nonlinear}. 
In particular, we write the Jacobian as a linear combination of basis-functions 
$A(x,u)=\sum_{i=1}^{n_\theta}A_{\theta,i} \theta_i(x,u),~A_{\theta,i}\in\mathbb{R}^{n\times n}$ 
with $\theta(x,u)=[\theta_1(x,u),\dots,\theta_{n_\theta}(x,u)]^\top\in\mathbb{R}^{n_\theta}$. Similar embedings are used for $h_{\mathrm{x},j}(x)$. 
Then, we determine a polytope $\Theta\subseteq\mathbb{R}^{n_\theta}$, such that $\theta(x,u)\in\Theta$ $\forall (x,u)\in\mathbb{R}^n\times\mathbb{U}$. 
Thus, we can replace $\forall (x,u)\in\mathbb{R}^n\times\mathbb{U}$ in Inequalities~\eqref{eq:CCM_LMI} with the sufficient condition $\forall \theta\in\Theta$. 
Since $\theta$ appears linearly in in Problem~\eqref{eq:CCM_LMI}, it suffices to verify the LMIs on the vertices of the set $\Theta$, thus reducing the problem to a finite-dimensional SDP.
For further details, see the theoretical derivations in~\cite[Prop.~1]{koehler2020nonlinear} and the online available code.
\subsubsection*{State-dependent contraction metric}
To address state-dependent contraction metrics $M(x)$, two modifications are required.
First, a finite parametrization is needed, typically of the form $W(x)=\sum_{i}W_{\theta_{{W}},i}\theta_{{W},i}(x)$ with basis functions $\theta_{{W},i}:\mathbb{R}^n\rightarrow\mathbb{R}$. 
Given that uniform bounds~\eqref{eq:CCM_M_bounds} need to hold globally for all $x\in\mathbb{R}^n$, the basis functions $\theta_{W,i}(x)$ should be globally bounded, such as sigmoid functions. 
Condition~\eqref{eq:CCM_anhiling} restricts the parametrization to only use component of the state $x$ that are not in the image of 
$G\in\mathbb{R}^{n\times q}$. 
This is achieved by computing a  matrix $\bar{G}$ that is orthogonal to $G$ and then considering a parametrization of the form $M(x)=\tilde{M}(\bar{G}x)$. 
To ensure satisfaction of~\eqref{eq:CCM_rho} for all $(x,u)\in\mathbb{R}^n\times\mathbb{U}$ extra care is required for the term $M(f(x,u,0))$. 
A typical approach is to derive a hyperbox $\Omega$, such that $\theta_{W}(f(x,u,0))\in\{\theta_{W}(x))\}\oplus\Omega$. 
Then, the convex embeddings can again be utilized to formulate the problem as a finite-dimensional SDP~\cite{koehler2020nonlinear}.

\subsection{PRS design for dynamics non-affine in $w$}
\label{app:affine_w}
In the following, we extend the PRS construction from  Section~\ref{sec:CCM} to address the case where $f$ is non-affine in $w$. 
This extension leverages an augmented state $\xi(k)=[x(k);w(k)]\in\mathbb{R}^{n_\xi}$, $n_\xi=n+q$, with noise $w_\xi(k)=w(k+1)\in\mathbb{R}^q$,  dynamics 
\begin{align*}
\xi(k+1)=&[f(x(k),u(k),w(k)); w(k+1)]\\
=:&f_{\xi}(\xi(k),w_\xi(k)),
\end{align*}
and Jacobian 
\begin{align*}
A_{\xi}(x,w,u):=&
\left.\begin{pmatrix}\dfrac{\partial f}{\partial x}&\dfrac{\partial f}{\partial w}\\ 0&0\end{pmatrix}\right|_{(x,w,u)}, ~
G_{\xi}:=\dfrac{\partial f_\xi}{\partial w_\xi}=\begin{pmatrix}0\\I_q\end{pmatrix}.
\end{align*}
\begin{theorem}
\label{thm:CCM_nonaffine}
Let Assumptions~\ref{asm:disturbance} hold and suppose $f$ is continuously differentiable.  
Suppose there exist a state-dependent matrix $M:\mathbb{R}^{n}\rightarrow\mathbb{R}^{n_\xi\times n_\xi}$, positive definite matrices $\underline{M},\overline{M}\in\mathbb{R}^{n_\xi\times n_\xi}$, a contraction rate $\rho\in[0,1)$, and a constant $\bar{w}\geq 0$, such that the following conditions hold for all $x\in\mathbb{R}^{n}$, $w\in\mathbb{R}^q$, $u\in\mathbb{U}$:
\begin{subequations}
\begin{align}
\label{eq:CCM_cond_aug_1}
 &\underline{M}\preceq M(x)\preceq \overline{M},\\
\label{eq:CCM_cond_aug_2}
& A_\xi(x,w,u)^\top M(f(x,u,w)) A_\xi(x,w,u)\preceq \rho M(x),\\
\label{eq:CCM_cond_aug_3}
&\tr{\Sigma_{\mathrm{w}}G_\xi^\top \bar{M} G_\xi}\leq \bar{w}.
\end{align}
\end{subequations}
Then, Assumptions~\ref{asm:PRS}--\ref{asm:PRS_nom} hold with
\begin{align}
\label{eq:PRS_design_augmented}
&\mathbb{D}_k=\left\{x|~\|E\cdot x\|_{\underline{M}}^2\leq\sigma_{\mathrm{x},k}\right\},\quad 
\sigma_{\mathrm{x},k}:=\dfrac{1-\rho^k}{1-\rho}\dfrac{\bar{w}\cdot \rho}{1-p}\nonumber\\
&\mathcal{R}_k(x_0,u(0:k-1))=z(k)\oplus\mathbb{D}_k,~k\in\mathbb{I}_{\geq 0}, 
\end{align}
$z(k)$ according to~\eqref{eq:nominal} and the lifting $E=[I_n;~0]\in\mathbb{R}^{n_\xi\times n}$.
\end{theorem}
\begin{proof}
\textbf{Part I: }
Analogous to Theorem~\ref{thm:CCM},  the incremental Lyapunov function $V_\delta:\mathbb{R}^{n_\xi}\times\mathbb{R}^{n_\xi}\rightarrow\mathbb{R}$ induced by the metric $M$ satisfies
\begin{subequations}
\begin{align}
\label{eq:CCM_augmented_1}
\|\xi-\tilde{\xi}\|_{\underline{M}}^2\leq V_\delta(\xi,\tilde{\xi})\leq &\|\xi-\tilde{\xi}\|_{\overline{M}}^2,\\
\label{eq:CCM_augmented_2}
\mathbb{E}_{w_\xi(k)}\left[V_\delta(f_\xi(\xi,w_\xi(k))),f_\xi(\tilde{\xi},0))\right]\leq &\rho V_\delta(\xi,\tilde{\xi})+\bar{w},
\end{align}
\end{subequations}
for any $\xi,\tilde{\xi}\in\mathbb{R}^{n_\xi}$, $u\in\mathbb{U}$, $k\in\mathbb{I}_{\geq 0}$. 
In particular, the augmented model $f_\xi$ is linear in $w_\xi$, Condition~\eqref{eq:CCM_anhiling} holds with $M(x)$ independent of $w$, and  $w_\xi(k)=w(k+1)$ also satisfies the conditions in Assumption~\ref{asm:disturbance}. \\
\textbf{Part II: } 
Let us denote $\xi(k)=[x(k);w(k)]$  and $\xi^z(k)=[z(k);0]$ with the energy $\delta(k)=V_\delta(\xi(k),\xi^{\mathrm{z}}(k))$. 
The initial condition satisfies
\begin{align}
\label{eq:CCM_augmented_3}
\mathbb{E}_{w(0)}[\delta(0)]\stackrel{\eqref{eq:CCM_augmented_1}}{\leq} \mathbb{E}_{w(0)}[\|G_\xi w(0)\|_{\overline{M}}^2]\stackrel{\eqref{eq:CCM_cond_aug_3}}{\leq} \bar{w}.
\end{align}
Using the lifting $E$, we get:
\begin{align}
\label{eq:CCM_augmented_4}
&\|E\cdot( x(k+1)-z(k+1))\|_{\underline{M}}^2\\
\stackrel{\eqref{eq:CCM_augmented_1}}{\leq}& V_\delta([x(k+1);0],[z(k+1),0])
\stackrel{\eqref{eq:CCM_cond_aug_3}}{\leq} \rho \delta(k),\nonumber
\end{align}
where the last inequality is analogous to~\eqref{eq:CCM_augmented_2}, but without $\bar{w}$ due to the absence of $w_\xi(k)$ on the left hand side.
Lastly, we follow the step from Theorem~\ref{thm:PRS} to get
\begin{align*}
&\mathbb{E}_{w(0:k)}[E\cdot (x(k+1)-z(k+1))\|_{\underline{M}}^2]
\stackrel{\eqref{eq:CCM_augmented_4}}{\leq} \rho\cdot\mathbb{E}_{w(0:k)}[\delta(k)]\\
\leq& \rho\cdot\mathbb{E}_{w(0:k-1)}[\mathbb{E}_{w(k)}[\delta(k)|w(0:k-1)]]\\
\stackrel{\eqref{eq:CCM_augmented_2}}{\leq} &\rho\cdot \mathbb{E}_{w(0:k-1)}[\rho\delta(k)+\bar{w}]\leq \dots\\
\leq& \rho\left(\rho^k \mathbb{E}_{w(0)}[\delta(0)]+\sum_{i=0}^k\rho^i \bar{w}\right)\\
\stackrel{\eqref{eq:CCM_augmented_3}}{\leq}& \rho\left(\rho^k\bar{w}+\sum_{i=0}^{k-1}\rho^i \bar{w}\right)=\rho\dfrac{1-\rho^{k+1}}{1-\rho}\bar{w}.
\end{align*}
Applying the Markov inequality finishes the proof. 
\end{proof}

\bibliographystyle{IEEEtran}  
\bibliography{Bibliography} 
%    
%!TEX root = ./main.tex
%%%%%%%%%%%%%%%%%%%%%%%%%%%%%%%%%%%%%%%%%%%%%%%%%%%%%%%%%%%%%%%%%%%%%%%%%%%%%%% 
 
\begin{IEEEbiography}[{\includegraphics[width=1in,height=1.25in,clip,keepaspectratio]{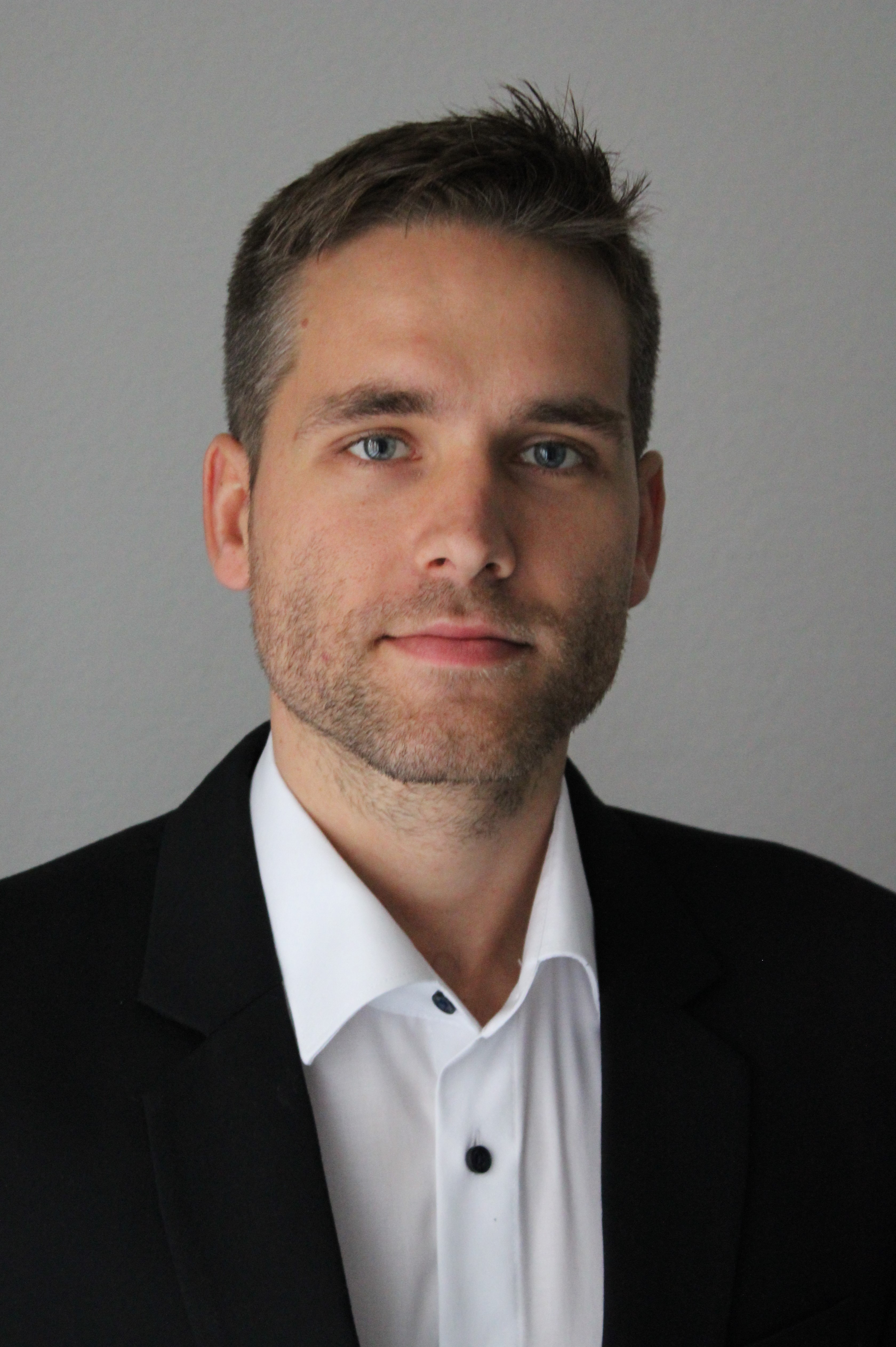}}]{Johannes K\"ohler}
received the Ph.D. degree from the University of Stuttgart, Germany, in 2021. He is currently a postdoctoral researcher at ETH Zürich, Switzerland. 
He is the recipient of the 2021 European Systems \& Control
PhD Thesis Award, the IEEE CSS George S. Axelby
Outstanding Paper Award 2022, and the Journal of
Process Control Paper Award 2023.
His research interests are in the area of model predictive control and control and estimation for nonlinear uncertain systems. 
\end{IEEEbiography}

\begin{IEEEbiography}[{\includegraphics[width=1in,height=1.25in,clip,keepaspectratio]{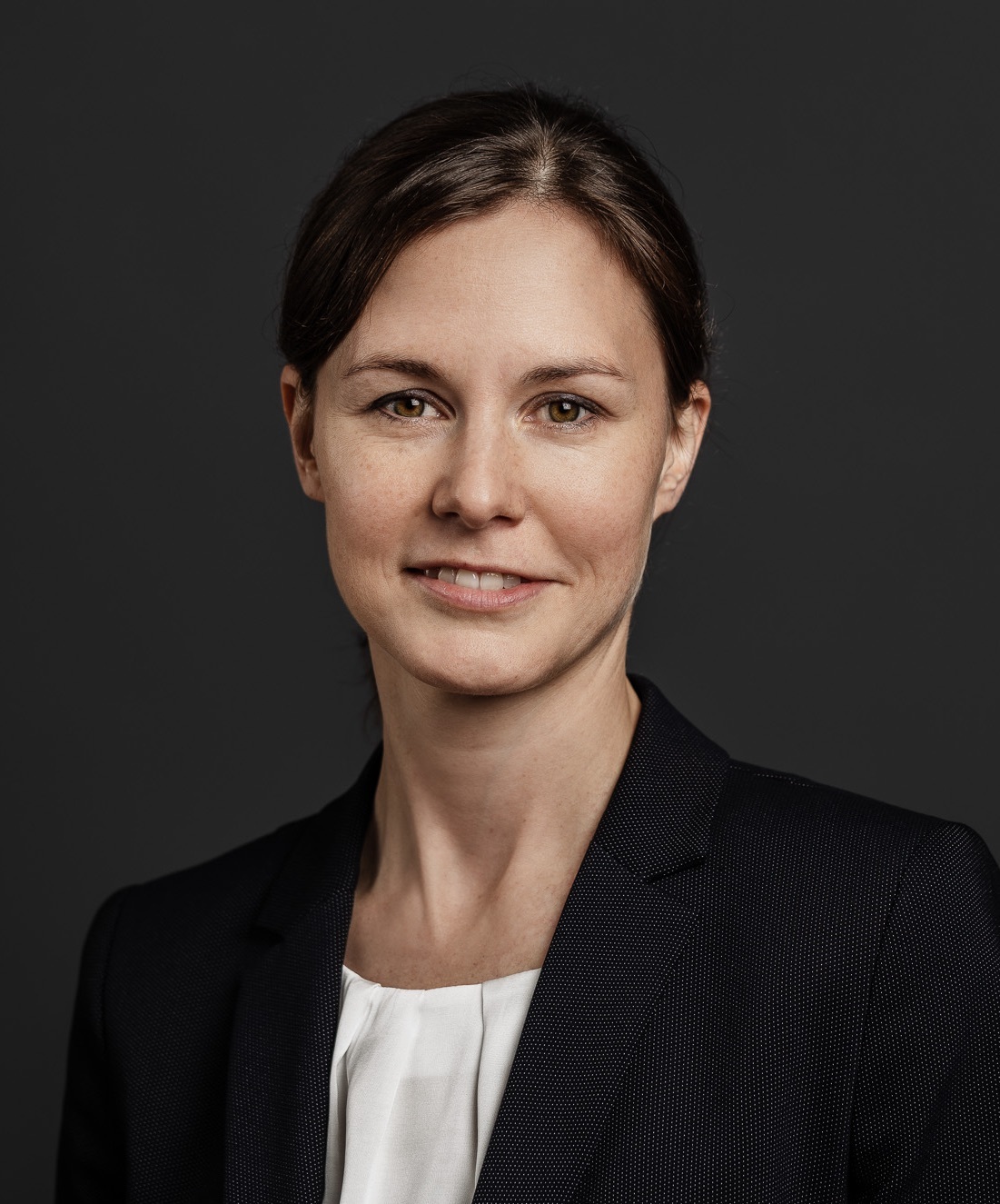}}]{Melanie N. Zeilinger}
 is an Associate Professor at ETH Zürich, Switzerland. 
She received the Diploma degree in engineering cybernetics from the University of Stuttgart, Germany, in 2006, and the Ph.D. degree with honors in electrical engineering from ETH Zürich, Switzerland, in 2011. 
From 2011 to 2012 she was a Postdoctoral Fellow with the Ecole Polytechnique Federale de Lausanne (EPFL), Switzerland.
She was a Marie Curie Fellow and Postdoctoral Researcher with the Max Planck Institute for Intelligent
Systems, Tübingen, Germany until 2015 and with the Department of Electrical Engineering and Computer Sciences at the University
of California at Berkeley, CA, USA, from 2012 to 2014. 
From 2018 to 2019 she was a professor at the University of Freiburg, Germany. 
Her current research interests include safe learning-based control, as well as distributed control and optimization, with applications to robotics and human-in-the loop control.
\end{IEEEbiography}
 
\end{document}